\documentclass{article}
\usepackage[margin=1.2in]{geometry}
\usepackage[utf8]{inputenc}

\usepackage[english]{babel}

\usepackage[a-2u]{pdfx}

\usepackage{lmodern}

\usepackage{amsmath}        
\usepackage{amsfonts}       
\usepackage{amsthm}         
\usepackage{amssymb}
\usepackage{bbding}         
\usepackage{bm}             
\usepackage{bbm}
\usepackage{stackrel}
\usepackage{graphicx}       
\usepackage{fancyvrb}       
\usepackage[style=alphabetic, sorting=nyt,
maxcitenames=99,
mincitenames=99,
minbibnames=99,
maxbibnames=99,
minalphanames=5,
maxalphanames=10,
backend=biber
]{biblatex} 
\usepackage[nottoc]{tocbibind} 
\usepackage{dcolumn}        
\usepackage{booktabs}       
\usepackage{paralist}       
\usepackage{enumitem}

\usepackage{todonotes}
\usepackage{csquotes}
\usepackage{dirtytalk}
\usepackage{relsize}
\usepackage{mathtools}
\usepackage{xfrac}
\usepackage{tikz}
\usepackage{tikz-cd}
\tikzcdset{every label/.append style = {font = \small}}
\usepackage{exscale}
\usepackage{xcolor}         

\mathtoolsset{showonlyrefs}  

\usetikzlibrary{arrows.meta}
\tikzcdset{arrow style=tikz}

\title{Lagrangian Relations and Quantum $L_{\infty}$ Algebras}

\author{Branislav Jurčo$^{1,a}$, Ján Pulmann$^{2,b}$, and Martin Zika$^{1,c}$\\
        \small $^{1}$Mathematical Institute, Faculty of Mathematics and Physics, Charles University, \\
        \small Prague, Czech Republic \\
        \small $^{2}$School of Mathematics, University of Edinburgh, 
        \small Edinburgh, United Kingdom \\\\
        \small $^{a}$ \email{Branislav.Jurco@mff.cuni.cz}, $^{b}$ \email{Jan.Pulmann@gmail.com}, $^{c}$ \email{Martin.Zika@mff.cuni.cz}
}

\date{} 

\hypersetup{unicode}
\hypersetup{breaklinks=true}
\definecolor{hyperlink}{RGB}{11,0,128}
\hypersetup{colorlinks, allcolors=hyperlink, linktoc=all}

\makeatletter
\def\@makechapterhead#1{
  {\parindent \z@ \raggedright \normalfont
   \Huge\bfseries \thechapter. #1
   \par\nobreak
   \vskip 20\p@
}}
\def\@makeschapterhead#1{
  {\parindent \z@ \raggedright \normalfont
   \Huge\bfseries #1
   \par\nobreak
   \vskip 20\p@
}}
\makeatother

\newtheoremstyle{theoremnew}
  {.5\baselineskip\@plus.2\baselineskip\@minus.2\baselineskip}
  {.5\baselineskip\@plus.2\baselineskip\@minus.2\baselineskip}
  {\slshape}
  {}
  {\bfseries}
  {.}
  { }
  {}

\theoremstyle{plain}
\newtheorem{theorem}{Theorem}[section]
\newtheorem{prop}[theorem]{Proposition}
\newtheorem{lemma}[theorem]{Lemma}

\newtheorem{cor}[theorem]{Corollary}
\newtheorem{remark}[theorem]{Remark}

\theoremstyle{remark}

\theoremstyle{definition}
\newtheorem{newdefn}[theorem]{Definition} 
\newenvironment{definition}
{\renewcommand{\qedsymbol}{$\blacktriangle$}%
\pushQED{\qed}\begin{newdefn}}
{\popQED\end{newdefn}}
\newtheorem{example}[theorem]{Example}

\renewcommand\qedsymbol{$\blacksquare$}

\definecolor{hyperlink}{RGB}{11,0,128}
\hypersetup{colorlinks, allcolors=hyperlink, linktoc=all}

\makeatletter
\newdimen\scalemath@axis
\newcommand*{\scalemath}[3]{%
  #1{%
    \mathpalette{\scalemath@aux{#2}}{#3}%
  }%
}
\newcommand*{\scalemath@aux}[3]{%
  \begingroup
    \everyvbox{}%
    \settoheight\scalemath@axis{$#2\vcenter{}$}%
    \raisebox{\scalemath@axis}{%
      \scalebox{#1}{%
        \raisebox{-\scalemath@axis}{%
          $\m@th#2#3$%
        }%
      }%
    }%
  \endgroup
}
\makeatother 

\newcommand*{\smallin}{\scalemath{\mathrel}{.75}{\in}}

\newcommand*{\smallsubset}{\scalemath{\mathrel}{.75}{\subset}}

\newcommand{\email}[1]{\href{mailto:#1}{#1}}

\definecolor{MyRed}{rgb}{0.75, 0.00, 0.2}
\definecolor{MyBlue}{rgb}{0.4, 0.1, 0.70}
\definecolor{JanColor}{rgb}{0.07, 0.50, 0.00}
\definecolor{MartinColor}{rgb}{0.00, 0.31, 0.75}
\definecolor{BranoColor}{rgb}{8.00, 0.60, 0.00}
\definecolor{Gray}{HTML}{888888}

\newcommand{\red}[1]{\color{black}{#1}\color{black}}

         \DeclareMathAlphabet{\mathscr}{U}{BOONDOX-cal}{m}{n}
         \SetMathAlphabet{\mathscr}{bold}{U}{BOONDOX-cal}{b}{n}
         \DeclareMathAlphabet{\mathbscr} {U}{BOONDOX-cal}{b}{n}
         \DeclareMathAlphabet{\mathpzc}{OT1}{pzc}{m}{it}

\DeclareMathAlphabet\mathbfcal{OMS}{cmsy}{b}{n}
\newcommand{\mathcat}[1]{\mathsf{#1}}

 \newcommand{\shift}[2]{ #2[{#1}]   }
 
 \newcommand{\s}[1]{\shiftbyone{#1}}
 \newcommand{\sT}{ \s{T^*} }

  \newcommand{\restr}[2]{{
  \left.\kern-\nulldelimiterspace 
  #1 
  \vphantom{\big|} 
  \right|_{#2} 
  }}
  
  \newcommand{\restrsmall}[2]{{
  \left.\kern-\nulldelimiterspace 
  #1 
  \vphantom{|} 
  \right|_{#2} 
  }}

\newcommand{\B}{\operatorname{\mathbf{B}}}

\newcommand{\D}[1]{\mathcal{D}^{_{\frac12}} \mkern-2mu #1}

\newcommand{\F}[1]{\mathcal{F}  #1 }
\newcommand{\Fw}[1]{\mathcal{F}_w  #1 }

\newcommand{\I}{\mathcal{I}}

\renewcommand{\L}{\mathcal{L}}
\newcommand{\M}{\mathcal{M} }

\newcommand{\Real}{\mathbb{R}}
\newcommand{\R}{\mathbb{R}}

\newcommand{\Z}{\mathbb{Z}}
 
\newcommand{\Red}{\mathbf{R}}
\newcommand{\Comp}{\mathbf{X}}

\newcommand{\GrVect}{\mathcat{GrVect}}
\newcommand{\OSC}{\mathcat{LinSymp_{-1}}}
\newcommand{\LinOSC}{\OSC}
\newcommand{\Coiso}{\mathcat{LinCoiso_{-1}}}
\newcommand{\QOSC}{\mathcat{LinQSymp_{-1}}}

\newcommand{\LinRed}{\mathcat{Red}_{-1}}
\newcommand{\Cospan}{\mathcat{CospanRed}_{-1}}

\newcommand{\Vect}{ \mathcat{Vect}}

\newcommand{\intover}[1]{\int\limits_{\mathrlap{#1}}}
\newcommand{\intBV}[2]{\intover{#1} \!  e^{ \Sfr{#2} / \hbar}}

\newcommand{\Linfty}{L_{\infty}}
\newcommand{\BV}{\boldsymbol{\Delta}}

\newcommand{\Ber}[1]{\operatorname{Ber} \left(#1\right)}
\newcommand{\Sfree}{{S_{\mathrm{free}}}}
\newcommand{\Sfr}[1]{{S_{\mathrm{free}}^{#1}}}
\newcommand{\SfreeR}{{S^R_{\mathrm{free}}}}
\newcommand{\Sint}{S_{\mathrm{int}}}

\newcommand{\dimsum}[2]{\operatorname{D}_{#1}  ( #2 ) }
\newcommand{\Ann}[1]{\operatorname{Ann} ( #1 )}

\newcommand{\lpartial}[2]{\frac{\partial_L #1}{\partial #2}}
\newcommand{\rpartial}[2]{\frac{\partial_R #1}{\partial #2}}
\newcommand{\sign}[1]{(-1)^{#1}}

\newcommand{\can}{\mathrm{can}}

\newcommand{\base}[1]{\mathbf{#1}}
\newcommand{\diag}[1]{\operatorname{\mathrm{diag}}( #1 )}
\newcommand{\even}{_{\mathrm{even}}} 
\newcommand{\odd}{_{\mathrm{odd}}} 
\renewcommand{\Im}{\operatorname{Im}}
\renewcommand{\ker}{\operatorname{Ker}}
\newcommand{\Ker}{\operatorname{Ker}}

\newcommand{\graph}[1]{\operatorname{Gr}_{#1}}

\newcommand{\id}{\mathbbm{1}}

\newcommand{\define}{ \coloneqq }

\newcommand{\Span}[1]{ \langle #1 \rangle}

\renewcommand{\deg}[1]{ \left\vert #1 \right\vert }
\newcommand{\after}{\circ}

\newcommand{\flip}[1]{\overline{#1}}

\renewcommand{\tilde}[1]{\widetilde{#1}}

\newcommand{\Rhbar}{\mathbb{R}((\hbar))}
\newcommand{\Sym}{\operatorname{Sym}}

\newcommand{\Dens}[1]{\left\vert #1 \right\vert}

\newcommand{\DensWeight}[2]{\Dens{#2}^{#1}}
\newcommand{\HalfDens}[1]{\Dens{#1}^{\frac12}}

\newcommand\dhxrightarrow[1]{%
\xrightarrow{#1}\mathrel{\mkern-14mu}\rightarrow
}

\newcommand\dhxleftarrow[1]{%
\leftarrow\mathrel{\mkern-14mu}\xleftarrow{#1}
}

\begin{document}

\maketitle

\begin{abstract}
\noindent 
Quantum $\Linfty$ algebras are higher loop generalizations of cyclic $\Linfty$ algebras. Motivated by the problem of defining morphisms between such algebras, we construct a linear category of $(-1)$-shifted symplectic vector spaces and distributional half-densities, originally proposed by Ševera. Morphisms in this category can be given both by formal half-densities and Lagrangian relations; we prove that the composition of such morphisms recovers the construction of homotopy transfer of quantum $\Linfty$ algebras. Finally, using this category, we propose a new notion of a relation between quantum $\Linfty$ algebras. 
\end{abstract}



\tableofcontents

\section{Introduction}

A symplectic vector space $(V, \omega_V)$ is a vector space $V$ equipped with an antisymmetric and non-degenerate pairing $\omega_V$. A natural notion of a morphism $f\colon (V, \omega_V) \to (W, \omega_W)$ between symplectic vector spaces is a linear map $f\colon V \to W$ such that $\omega_V(v, v') = \omega_W(f(v), f(v'))$. However, this condition immediately forces $f$ to be injective.\footnote{For $v\in \Ker{f}$, we get $\omega_V(v, v') = \omega_W(f(v), f(v'))=0$. Requiring instead that $f$ preserves the inverse of the pairings, we get that $f$ is surjective.} Thus, the naive notion of a category of symplectic vector spaces is rather restrictive. A common solution to this issue is to replace morphisms $f\colon V \to W$ by Lagrangian subspaces of $(V, -\omega_V)\times (W, \omega_W)$, famously advocated by Weinstein and Guillemin-Sternberg \cite{gs:integralgeometry, weinstein:symplecticgeometry, weinstein:symplectic_categories, gs:semi-classical} and inspired by the work of H\"ormander \cite{Hoermander1971}. Graphs of symplectic isomorphisms $V\to W$ give examples of such Lagrangian subspaces.
\smallskip

We would like to consider symplectic vector spaces equipped with an additional algebraic structure, that of a \emph{quantum $\Linfty$ algebra}. These are homotopy and higher loop generalizations of graded Lie algebras equipped with a compatible degree $-1$ symplectic form. They first appeared in string field theory \cite{zwiebach:closed_string}, and can be succintly described using the Batalin-Vilkovisky formalism \cite{bv, markl2001loop, Barannikov2010, doubek-jurco-munster}: a quantum $L_\infty$ algebra on such $(-1)$-shifted symplectic vector space $(V,  \omega)$ is given by a formal power series $S \in\operatorname{Sym} (V^*)[[\hbar]]$ satisfying the quantum master equation
\begin{equation} \BV e^{S/\hbar}  = 0, \end{equation}
where the Batalin-Vilkovisky operator $\BV$ is defined using the degree $-1$ symplectic form.
\smallskip

To combine Lagrangian relations and quantum $\Linfty$ algebras we need a further enlargement of the $(-1)$-shifted symplectic category, proposed by \v{S}evera \cite{severa:qosc}. Morphisms from $V_1$ to $V_2$ in this \emph{quantum $(-1)$-shifted symplectic category} are ``distributional'' half-densities on $V_1\times V_2$, with a Lagrangian relation $L\subset V_1 \times V_2$ seen as $\delta$-like half-density supported on $L$. The Batalin-Vilkovisky Laplacian $\BV$ acts naturally on half-densities, and the composition, given by integration along the common factor, is compatible with $\BV$. In this setting, the quantum $\Linfty$ algebra $S$ can be encoded by a $\BV$-closed morphism from a point $*$ to $V$, given by the half-density $e^{S/\hbar}\sqrt{dV}$ on $*\times V \cong V$.
\medskip

In this paper, we rigorously define a natural class of distributional half-densities on $(-1)$-shifted symplectic vector spaces, which we call \emph{generalized Lagrangians}.\footnote{See also Remark \ref{rmk:distributional_halfdensity}, which further explains why it is natural to see (distributional) half-densities as a generalization of Lagrangian subspaces.} For two such half-densities on $V_1\times V_2$ and $V_2 \times V_3$, we define a Batalin-Vilkovisky integration along $V_2$, giving a partially defined composition\footnote{The composition is defined if we can compute the perturbative Gaussian integral, i.e.\ if a relevant quadratic form is non-degenerate.} on our version of the quantum $(-1)$-shifted symplectic category. Our main application and the original motivation for our work comes from composing the quantum $\Linfty$ algebra $e^{S/\hbar}\sqrt{dV} \colon * \to V$ with a surjective Lagrangian relation $L\colon V \dhxrightarrow{} W$. As both of these distributional half-densities are $\BV$-closed, their composition is $\BV$-closed as well and defines a quantum $\Linfty$ algebra on $W$. This construction goes back to Losev (see e.g.\ references in \cite{cattaneo_mnev:chern-simons_invariants}) and is known as the BV pushforward, homotopy transfer, or the effective action; it was later examined by many authors \cite{cattaneo_mnev:chern-simons_invariants, cattaneo_felder:effective_BV, costello, mnev, ChuangLazarevMinModel, BarannikovSolving, cattaneo_mnev_reshetikhin:perturbative_gauge_theories, braun_maunder:minimal_models,  doubek_jurco_pulmann:quantum_L_infty_and_HPL}. 

\subsection{Content of the paper}
In the second section, we recall some useful facts about Lagrangian relations of $(-1)$-shifted symplectic vector spaces. The content of this section is mostly standard, with many results adapted from the book of Guillemin and Sternberg \cite{gs:semi-classical} to the $(-1)$-shifted setting. We emphasize the canonical factorization of a Lagrangian relation into a reduction and a coreduction. Our results on composition of these factorizations in Sections \ref{sec:double_reductions}, \ref{sec:comp_and_factor} appear to be new.
\smallskip

In the third section, we introduce formal half-densities and formal Batalin-Vilkovisky fiber integrals along surjective Lagrangian relations. The linear $(-1)$-shifted symplectic category is a natural setting for these integrals, providing an invariant way to define fiber BV integration \cite{schwarz:geometry_of_bv, albert_bleile_frohlich:bv_integrals}. Finally, we relate this version of the BV fiber integral to the homological perturbation lemma.
\smallskip

In the fourth section, we start by defining linear distributional half-densities, called \emph{generalized Lagrangians}, on a $(-1)$-shifted symplectic vector space $V$. Roughly, they are given by a \emph{coisotropic subspace} $C\subset V$ and a formal half-density on the coisotropic reduction $C/C^\omega$. Using this generalized notion of a Lagrangian relation, we define a version of the quantum $(-1)$-shifted symplectic category $\QOSC$. The composition is defined using the fiber BV integral along a reduction constructed from the coisotropic relations, and we give some examples of such compositions. We finish by interpreting the construction of the effective action of \cite{doubek_jurco_pulmann:quantum_L_infty_and_HPL} as a commutative triangle in $\QOSC$, and proposing a more general symmetric relation between quantum $\Linfty$ algebras using factorization from Section \ref{sec:osc}.

\subsection{Related and future work}
We will now explain how our work relates to quantum field theory, point to other works studying Lagrangian and coisotropic correspondences, and list some directions of research.
\smallbreak

The category of non-linear (or smooth) Lagrangian relations has received considerable attention both from the viewpoint of symplectic geometry and mathematical physics \cite{weinstein:symplectic_categories, gs:semi-classical,cattaneo_mnev_reshetikhin:classical_BV, CHSaksz}.  Coisotropic relations, which appeared naturally in the present paper when considering distributional half-densities as in \cite{severa:qosc}, are less studied; see, however, \cite{weinstein:coiso_and_poisson_groupoids, SH:coiso} and \red{\cite{cattaneo_mnev:chern-simons_invariants} (see the next subsection for more details about the last reference).}\ Half-densities in the $0$-shifted symplectic  
\pagebreak

\noindent
   setting were extensively studied by Guillemin and Sternberg.\footnote{Let us mention a different way to add half-densities to  the linear symplectic category: The enhanced symplectic category introduced by Guillemin and Sternberg \cite{gs:semi-classical}, following H\"ormander \cite[{Ch.~IV}]{Hoermander1971}, has objects given by symplectic spaces and morphisms given by half-densities on Lagrangian relations \vspace{-1mm}$$(L \subset V\times W , \rho \in \HalfDens{L} ).$$
\noindent One can mimic this construction in the $(-1)$-shifted symplectic case, but quantum $\Linfty$ algebras do not induce any natural half-densities on a Lagrangian subspace, and the resulting Guillemin-Sternberg category of enhanced relations does not involve the integration theory we are looking for. We wish to describe BV fiber integration over the fiber $I$ of a coisotropic reduction (in the sense of Lemma \ref{lemma:weightsandBV}). But in the enhanced symplectic category, one encounters no natural densities on $I$ and composition is defined using solely the canonical isomorphisms from Lemmata \ref{lemma:density_properties}, \ref{lemma:short_exact_dens}.
For these reasons, we \emph{will not use} this similarly-looking construction.}

\subsubsection{Physics}
Lagrangian subspaces and half-densities in this work have their origin in the Batalin-Vilkovisky approach to quantum field theory. The $(-1)$-shifted symplectic vector space $V$ we consider should be seen as the (linear and finite-dimensional approximations of the) space of BRST fields and their antifields. The half-densities relevant to physics are of the form $e^{S/\hbar} \sqrt{dV}$ or $F e^{S/\hbar} \sqrt{dV}$ for an observable $F$. A Lagrangian subspace $L\subset V$ corresponds to choosing a gauge and pairing it with $F$ computes the expectation value of $F$
\begin{equation}\label{eq:physmot}
     \int_L  F e^{S/\hbar} \sqrt{dV}  = \langle F\rangle \int_L e^{S/\hbar} \sqrt{dV}.
\end{equation}

This interpretation can be extended to the category $\QOSC$. Its objects, $(-1)$-shifted symplectic
vector spaces, are spaces of field histories, while morphisms $V\to W$ are ``coupled'' quantum field theories on $V\times W$. The composition in $\QOSC$ is given in terms of the BV fiber integral; the expectation value \eqref{eq:physmot} can be seen as the result of the following composition
    \begin{equation} \label{eq:compositionrhoL}
        \begin{tikzcd}
            * \arrow[rrr, "{ F e^{S/\hbar} \sqrt{dV}}"] &&& V \arrow[rrr, "{\delta_L}"] &&& *
            \end{tikzcd}.
    \end{equation}

Notably, examples of generalized Lagrangian relations appeared in the works of Cattaneo, Mn\"ev and Reshetikhin. In \cite[Section~2.2.2, Remark~2.14]{cattaneo_mnev_reshetikhin:perturbative_gauge_theories}, the authors explain that a coisotropic subspace $\mathcal C \subset \mathcal F$ induces  ``BV pushforward'', a chain map from half-densities on $\mathcal F$ to half-densities on the coisotropic reduction $\underline{\mathcal C}$ of $\mathcal C$. This morphism is the (infinite-dimensional and non-linear) version of post-composition in $\QOSC$ with the Lagrangian relation $\mathcal F \to \underline{\mathcal C}$ given by the coisotropic reduction. 
\smallskip

Remarkably, in an earlier work \cite[Section~3.5] {cattaneo_mnev:chern-simons_invariants}, Cattaneo and Mn\"ev interpret the effective action calculated using a propagator constructed from a homotopy operator $K$ such that $K^2\neq 0$ as a homotopy transfer along ``a Gaussian-smeared Lagrangian subspace'' or, in other words, a ``thick BV integral''. In our terminology, the homotopy $K$ should induce a generalized Lagrangian with an underlying surjective coisotropic relation. We plan to explicitly describe these generalized Lagrangians in future work.

 We also expect to find other examples of generalized Lagrangian relations (distributional half-densities) in quantum field theory, for example coming from the AKSZ formalism \cite{aksz}. 
 \smallskip

 This categorical viewpoint also relates physics with homotopy algebras, see \cite{doubek_jurco_pulmann:quantum_L_infty_and_HPL, JuretalLinfty, homological_quantum_mechanics} and references therein, and \cite{GKW} for a recent highlight. Let us also mention that Lagrangian relations between $0$-shifted symplectic spaces appear in physics in many places; see e.g.\ \cite{bates_weinstein, cattaneo_mnev_reshetikhin:classical_BV,  arvanitakis:defects}.

\subsubsection{Linear logic}
The category of linear Lagrangian is extensively studied in linear logic. It has known presentations by generators and relations and is related to quantum computing, electrical circuits, and others; see \cite{ComfortKissinger2022} and the recent survey in \cite{BCC2024}.

Recently, the category of coisotropic relations was studied by Lorand and Weinstein \cite{LorandWeinstein2016}, and in the linear logic community by Booth, Carette, and Comfort \cite{BCC2024, comfort:stabilizer_codes}, although their physical motivation is different—coisotropic relations are related to \say{discarding} in quantum computing. It would be interesting to extend their approach to the $(-1)$-shifted case and graded coisotropic correspondences, as well as half-densities.

\subsubsection{Shifted symplectic geometry}
The work of Gwilliam and Haugseng \cite{gwilliam_haugseng:bv_functor} on linear BV quantization should be closely related to ours. They consider an $\infty$-category $\mathcal{Q}\text{uad}_1$ of vector spaces $(V, \omega)$ (or more generally modules over a cdga) with a degree $1$ pairing, where $1$-morphisms are given by linear maps $V\to V'$ together with a homotopy between $\omega$ and the pullback of $\omega'$. The truncation of $\mathcal{Q}\text{uad}_1$ to a 1-category, or its image under the $H_\infty$ functor, should be compared to our category $\QOSC$ (up to taking duals, to match conventions). However, at the moment, we do not understand e.g.\ how to get a generalized Lagrangian from the data of a $1$-morphism in $\mathcal{Q}\text{uad}_1$. 
\smallskip

In shifted geometric quantization, it was noted by Safronov \cite{SafronovShiftedGeometricQuantization} that the path integral pairing \eqref{eq:compositionrhoL} of an observable with a Lagrangian should be compared to the $(-1)$-shifted geometric quantization for a given prequantization and polarization. Concretely, the following diagram is compared to \eqref{eq:compositionrhoL} \cite{SafronovShiftedGeometricQuantization}  in
\[\begin{tikzcd}[column sep=1in]
	{*} & {\Omega_\omega(X)} & {*}
	\arrow["{\text{prequantization}}", from=1-1, to=1-2]
	\arrow["{\text{polarization}}", from=1-2, to=1-3]
\end{tikzcd}\]
with the middle object being ($\omega$-twisted differential forms, or half-densities on) a $(-1)$-shifted symplectic space. It is therefore natural to ask whether there exists a coisotropic generalization of polarizations.

\subsubsection{Homotopy transfer}
One can also understand the present work as providing an invariant geometric language for homotopy transfer (of quantum $L_\infty$ algebras). Special deformation retracts or abstract Hodge decompositions \cite{chuang_lazarev:hodge_decomposition}, are a basic object in the theory of homological perturbations\footnote{See \cite[Sec.~1]{doubek_jurco_pulmann:quantum_L_infty_and_HPL} for an overview of the history of homological perturbation theory.} \cite{crainic:hpl, markl:ideal_hpl}, as they can be used to transfer algebraic structures along homotopy equivalences. We prove in Proposition \ref{prop:SDRnondeg} that symplectic special deformation retracts are in bijection with non-degenerate reductions, an arguably more natural notion. It would be interesting to see if more of the theory of homotopy transfer has similar interpretation.

\subsubsection{Non-linear generalizations}
A natural generalization of the present category is to allow non-linear Lagrangian submanifolds of the product as morphisms between $(-1)$-shifted symplectic manifolds. This would allow for more general non-strict morphisms of quantum $\Linfty$ algebras, such as the non-linear symplectic diffeomoprhism constructed in \cite[Sec.~4.3.1]{doubek_jurco_pulmann:quantum_L_infty_and_HPL} which gives a homotopy to the effective action. Since we work perturbatively, it is natural to work with formal or \emph{micro} Lagrangian relations, as introduced in the work of Cattaneo, Dhenin and Weinstein \cite{cattaneo_dherin_weinstein:microsymp1}; in the BV context these were later studied by T. Voronov in \cite{Voronov2017} under the name \emph{thick morphisms}.

\subsubsection{Morphisms of quantum \texorpdfstring{$L_\infty$}{L-infinity} algebras}

Finally, there are other notions of morphisms of quantum $\Linfty$ algebras which we can encode using the linear category $\QOSC$. For example, the natural notion of an equivalence of quantum $\Linfty$ algebras, introduced by Mn\"ev in \cite[Def.~17]{mnev}, is closely related to our Definition \ref{def:rel_of_qLinfty}; our notion allows more general reductions that to homology, while Mn\"ev's allows equality-up-to-homotopy of the effective actions. The category $\QOSC$ also provides an answer to Mn\"ev's question posed in the remark below \cite[Def.~17]{mnev}: both induction and (in our setting linear) isomorphims can be seen as commutative triangles as in Eq. \eqref{eq:triangleunderpoint}. 
In future work, we would like to examine post-compositions with generalized Lagrangians not given by Lagrangian relations, as in the beginning of Section \ref{sec:relsqlinfty}, or construct spans as in Remark \ref{rmk:constructiongspans}. Moreover, we see a hint of a 2-categorical structure naturally appearing in Remark \ref{rmk:R_2functor}, it would be interesting to extend it further. Some of these constructions appear to have better properties in the non-linear setting. We plan to address this in a future work.

\subsection{Acknowledgements}
We would like to thank Pavol Ševera for explaining to us the quantum odd symplectic category and its relation to effective actions. We would also like to thank Owen Gwilliam and Rune Haugseng for discussions about their related work and Christian Saemann for his interest and useful discussions. J.P. would like to thank Cole Comfort and Fridrich Valach for answering his questions about Lagrangian relations, Justin Hilburn for suggestions leading to Remark \ref{rmk:distributional_halfdensity} and comments on shifted geometric quantization; and  Tudor Dimofte, Lukas M\"uller and Pavel Safronov for further helpful discussions. M.Z. would like to thank Roberto Zucchini for inspiring discussions. Finally, we would like to thank the anonymous referees for useful suggestions which improved the paper and supplied multiple missing references.

J.P. was supported by the Postdoc.Mobility grant 203065 of the SNSF. B.J. and M.Z. were supported by GA\v CR  grant EXPRO 19-28628X. B.J. was further supported by GA\v CR grant 24-10887S. M.Z. was supported by the GAUK 283723 grant and benefited from the SVV-2023-260721 project.

\section{Linear \texorpdfstring{$(-1)$}{(-1)}-Symplectic Category}\label{sec:osc}

We start by recalling some elementary definitions from graded linear algebra over the field $\Real$.
\medskip

A \textbf{graded vector space} $V$ is a direct sum of real vector spaces $V= \bigoplus_{i \in \Z} V_i$; \emph{we will always assume that $V$ is finite-dimesional.}\footnote{The results in Section \ref{sec:osc} hold with a weaker condition that $V$ is of finite type, i.e.\ $V_i$ is finite-dimensional for all $i$.} The zero-dimensional vector space will be denoted by $*$. Elements $v \in V_i\subset V$ are called homogeneous of degree $\deg{v}=i$.  The degree shift\footnote{This means that $\mathbb R[k]$ is concentrated in degree $-k$ and that $V\oplus\shift{-k}{V^*}$ will have a pairing of degree $-k$.} is denoted $\left( \shift{j} {V}  \right)_i =V_{i+j}$. The graded dual $V^*$ is defined to have a reflected degree: $\left( V^* \right)_i = \left( V_{-i}  \right)^* $. A linear map $f$ is said to have degree $k \in \Z$ if $\deg{f(v)}=\deg{v}+k$ for any homogeneous element $v \in V$. A morphism of graded vector spaces is a linear map of degree~0. Linear maps $V \to W $ of degree $k$ can be thought of as morphisms in $\GrVect \left(  V, \shift{k} W \right) $. A subspace of a graded vector space $W \subseteq V$ is a linear subspace embedded by a morphism of graded vector spaces. The annihilator of $W\subset V$ is a graded subspace $\Ann W \subset V^*$ with graded components
\begin{equation}
    (\Ann W)_k = \{ \alpha \in (V^*)_k \ \vert \ \restr{\alpha}{W} = 0\}.
\end{equation}
\noindent
Note that in the category $\GrVect$, short exact sequences are well-defined and always split. This is inherited from the category of finite-dimensional real vector spaces $\Vect$ degree-wise.

\begin{definition}
    Let $V$ be a finite-dimensional $\mathbb Z$-graded vector space. The \textbf{dimensional generating function} is defined as the Laurent polynomial\footnote{If we allow $V$ of finite type, $\dimsum{V}{s}$ is an element of $\mathbb{N}[[s,s^{-1}]]$.}
\[
    \dimsum{V}{s} \define  \sum_{k \in \Z} \left( \dim V_k \right) s^k  \in \mathbb{N} [s,s^{-1}]. \qedhere
\]
\vspace{-4mm}
\end{definition}
This object caries all the information of $V$ invariant under isomorphisms of graded vector spaces. Such notation is convenient for manipulations with degree shifts and degree reflections. \\

\begin{lemma}\label{lemma:dimsum} Let $V,W$ be graded vector spaces, $k \in \Z$. Then
\begin{enumerate}
    \item $\dimsum{V \times W}{s} = \dimsum{V}{s} + \dimsum{W}{s},$
    \item $\dimsum{\shift{k} V }{s} = s^{-k} \dimsum{V}{s},$
    \item $\dimsum{V^*}{s} = \dimsum{V}{s^{-1}}.$
    \item Given an invertible linear map $f: V \to W$ of degree $k$, i.e.\ equivalently an invertible morphism in $\GrVect \left(  V ,\shift{k} W \right) $, we have 
    \begin{equation}
         s^{k}  \dimsum{V}{s} = \dimsum{W}{s}.
    \end{equation}
    \vspace{-5mm}
    \item For $W\subset V$ a graded subspace, $\dimsum{V/W}{s} = \dimsum{V}{s} - \dimsum{W}{s}$.
\end{enumerate}
\end{lemma}

\subsection{\texorpdfstring{$(-1)$}{(-1)}-Shifted Symplectic Vector Spaces}
  We will focus on linear symplectic structures of degree $-1$ (also known as P-structures\footnote{This is an odd(-shifted) symplectic structure. Note that even(-shifted) symplectic structures have perhaps analogous but different behaviour; see \cite{roytenberg:derived_brackets} for their role in generalized geometry.}).

\begin{definition}
    A $(-1)$\textbf{-shifted symplectic vector space} or just $(-1)$\textbf{-symplectic vector space} is a graded vector space $V$ equipped with a non-degenerate graded-antisymmetric bilinear map $\omega: V \times V \to \Real$ of degree $\deg{\omega} = -1$. In other words, a bilinear map such that for all $v,w \in V$,
    \begin{enumerate}
        \item $\omega(v, w) \neq 0$ only if  $\deg{v}+ \deg{w} = 1$,
        \item $\omega(w, v) =  -(-1)^{\deg{w}\deg{v}}\omega(v, w) = -\omega(v, w).$
    \end{enumerate}
    An isomorphism of graded vector spaces $f \colon \left( V , \omega_V \right) \to \left( W , \omega_W \right)$ is said to be a \textbf{symplectic isomorphism} if $f^* \omega_W = \omega_V$.
\end{definition}

\begin{remark}\label{rmk:odd_symp_dim}
The existence of a symplectic structure with a non-zero degree imposes conditions on dimensionality of $V$. Since the map  $V \to V^*$ given by $x \mapsto \omega(x, -)$ is an isomorphism of degree $\deg{\omega} = -1 $, by Lemma \ref{lemma:dimsum} we have 
\begin{equation}
    s \dimsum{V}{s^{-1}} =  \dimsum{V}{s} .
\end{equation}
\end{remark}

\begin{example}\label{ex:shiftedcotangent}
    Define $\sT W \define \s{W^*}  \oplus W$, the \textbf{shifted cotangent bundle} of $W \in \GrVect$, with the cotangent fiber concentrated in degree $+1$. The canonical symplectic structure $\omega_{\mathrm{can}}$ given by
    \begin{equation}
        \omega_{\mathrm{can}}( \alpha\oplus v, \alpha'\oplus v' ) = \alpha(v') - \alpha'(v)
    \end{equation}
    is a $(-1)$-shifted symplectic structure. In case $W$ is purely even, the cotangent fibers are purely odd. In fact, every $(-1)$-shifted symplectic vector space $V$ is linearly symplectomorphic to such odd cotangent bundle. For example, we can choose $\smash{W = \bigoplus_{k\le 0} V_k}$ or $\smash{W = \bigoplus_{k \in \mathbb Z} V_{2k}}$ as the base, and $x \mapsto \omega(x, -)$ is a symplectic isomorphism between the remaining graded components of $V$ and $\s{W^*}$.

    \smallskip
    Schwarz \cite[Thm.~3]{schwarz:geometry_of_bv} proved a more general statement extending the setting to the category of supermanifolds. The idea is that since the odd directions are de Rham exact, a \emph{Moser path method argument} can be carried out to identify all odd symplectic structures with the canonical one up to symplectomorphism.
\end{example}

\begin{definition}[{see e.g.\ \cite[Sec.~2.1]{McDuffSalamon}}]
We define the \textbf{symplectic complement} of a subspace $W \subseteq \left( V , \omega \right)$ by \begin{equation}W^{\omega} = \{v\in V \mid \omega(v, w) = 0, \, \forall \, w\in W \}.
\end{equation}

\pagebreak
\noindent We say a subspace $W$ is 
    \begin{itemize}
        \item \textbf{isotropic} if $W \subseteq W^\omega $,
        \item \textbf{coisotropic} if $W^\omega \subseteq W $,
        \item \textbf{Lagrangian} if $W = W^\omega$,
        \item \textbf{symplectic} if $W \cap W^\omega = 0$. 
    \end{itemize}
     Equivalently, a subspace $W$ is symplectic if $\omega$ restricts to a non-degenerate pairing on $W$.
\end{definition}

\begin{example}
    Let $V$ be a graded vector space and $W\subset V$ a graded subspace. The shifted conormal bundle of $W$, given by
    $$\s{N^*} W = \s{\Ann W} \oplus W \subset \sT V = \s {V^*} \oplus V$$ is a Lagrangian subspace of $\sT V$ from Example \ref{ex:shiftedcotangent}. Note that in the setting of odd symplectic supermanifolds, Schwarz \cite[Thm.~4]{schwarz:geometry_of_bv} proved that any Lagrangian submanifold of $\sT M$ can be smoothly deformed into the shifted conormal bundle of a submanifold of $M$.
\end{example}

\begin{definition} \label{def:coiso_red}
    A \textbf{coisotropic reduction} of a coisotropic subspace $C$ is the quotient space $C/C^\omega$ together with the symplectic structure $\omega_R$ induced on $C/C^\omega$,
    \[
        \left( V , \omega \right) \stackrel{\iota}{\hookleftarrow} C \stackrel{\pi}{\to} \left( C/C^\omega , \omega_R \right), \quad \text{such that} \quad \iota^* \omega = \pi^* \omega_R . \qedhere
    \]
\end{definition}
    We will often omit $\iota$ and write $ \iota (c) \equiv c $ for $c \in C$. Let us record a simple but useful lemma from \cite[Lecture 3]{weinstein:lectures} which translates verbatim to the $(-1)$-shifted setting.

    \begin{lemma}\label{lemma:reduced_lagrangian}
        For $C\subset V$ coisotropic and $L\subset V$ Lagrangian, the image of $L\cap C$ in $C/C^\omega$, denoted $[L]_C$, is Lagrangian.
    \end{lemma}
    \begin{proof}
        $[L]_C^\omega$ is the image in $C/C^\omega$ of  $(L\cap C)^\omega \cap C =   (L + C^\omega ) \cap C = L \cap C + C^\omega$, where the last equality holds since $C^\omega \subset C$.
    \end{proof}

\begin{example}
    Let us consider a $(-1)$-shifted symplectic vector space $(V, \omega)$. A compatible\footnote{I.e.\ $\omega \colon V\otimes V \to \s \Real$ is a chain map.} differential is a differential $Q\colon V \to V$ such that $\omega(Qx, y) + (-1)^{\deg x} \omega(x, Qy) = 0$. Then $\operatorname{Im}Q$ is isotropic, as $\omega(Qx, Qy) = \pm \omega(x, Q^2 y) = 0$. Elements of the symplectic complement $v\in (\operatorname{Im}Q)^\omega$ have to satisfy, for any $x \in V $;
    $$ \omega(Qx, v) = 0,  \quad \text{equivalently} \quad  \omega(x, Qv) = 0,  \quad \text{or equivalently} \quad  Qv = 0.$$ This means that $(\operatorname{Im}Q)^\omega = \Ker{Q}$. Thus, the cohomology of $Q$ is also the coisotropic reduction of $\Ker Q$. In this example, the differential is zero when restricted to the isotropic subspace $\operatorname{Im}Q$; we will be mostly interested in isotropic subspaces $I\subset V$ such that $\Ker{Q}\cap I = \{0\}$, i.e.\ the opposite situation.
\end{example}

Similarly to the classical case, we can equivalently describe a Lagrangian subspace as a (co)isotrope with appropriate dimensionality. Note that in the graded case, only the sums $\dim L_k + \dim L_{-k+1}$ are determined for Lagrangian $L$.

\begin{lemma}\label{lemma:dimsum_symp}
    Let $W$ be a graded subspace of a $(-1)$-shifted symplectic vector space $\left( V , \omega \right)$. Then
    \begin{equation}\label{eq:Dimsumcomplement}
         s^{-1} \dimsum{W^\omega}{s} =  \dimsum{V}{s^{-1}} - \dimsum{W}{s^{-1}} 
    \end{equation}
    and in particular, an isotropic (or coisotropic) subspace $L\subset V$ is Lagrangian if and only if
    \begin{equation}
        s^{-1} \dimsum{L}{s} =  \dimsum{V}{s^{-1}} - \dimsum{L}{s^{-1}} .
    \end{equation}
\end{lemma}
\begin{proof}
    The map $x\mapsto \omega(x, -)$ restricts to an isomorphism $W^\omega \stackrel{\sim}{\to} \Ann{W}$ of degree $-1$, thus by Lemma \ref{lemma:dimsum} we have 
    \begin{equation}
        \dimsum{\Ann{W}}{s} = s^{-1}\dimsum{W^\omega}{s} .
    \end{equation}
    Finally, from $W^* \cong V^*/\Ann{W}$, we have $\dimsum{V}{s^{-1}} = \dimsum{W}{s^{-1}} + \dimsum{\Ann{W}}{s}$ and equation \eqref{eq:Dimsumcomplement} follows. For the second part of the lemma, $L \subseteq L^\omega$ (or $L^\omega \subseteq L $) and $\dimsum{L}{s}=\dimsum{L^\omega}{s}$ together imply $L=L^\omega$.
\end{proof}

\begin{lemma}\label{lemma:ortho_comp_squared}
    $ \left( W^\omega \right)^\omega = W $
\end{lemma}
\begin{proof}
    Clearly, $W \subseteq \left( W^\omega \right)^\omega $. By double application of Lemma \ref{lemma:dimsum_symp} and Remark \ref{rmk:odd_symp_dim}, $\dimsum{\left( W^\omega \right)^\omega}{s} =  \dimsum{W}{s}$. Together, these observations yield the statement.
\end{proof}

\subsubsection{Coisotropes and Non-canonical Decompositions}\label{ssec:coisotropesdec}

Given a coisotropic subspace $C \subset V$, the quotient $C/C^\omega$ has again a natural degree $-1$ symplectic form as in Definition \ref{def:coiso_red}. We will now show that $V$ is isomorphic to $C/C^\omega\oplus \sT{C^\omega}$ in a non-canonical way. For $C$ Lagrangian, this implies that Lagrangian complements always exist.

\begin{prop}\label{prop:red_decomp}
    Let $C \subseteq \left( V, \omega \right) $ be a coisotropic subspace. Denote $I \define C^{\omega} \subseteq C$ its isotropic complement. 
    Then, there exist complements\footnote{$R$ stands for \emph{reduced}, $B$ stands for \emph{boundaries}. This is motivated by the canonical decomposition from Section \ref{ssec:canon_decomp}; the subspace $B$ will consist of (co)boundaries of a differential.} $B\subset V$ of $C$ and $R \subset C$ of $I$ such that
    \begin{enumerate}
        \item $R$ and $R^{\omega} = I \oplus B$ are symplectic subspaces of $V$,
        \item $I$, $B$ are Lagrangian subspaces of $I \oplus B$.
    \end{enumerate} 
    In other words, we have a (non-canonical) direct sum decomposition
    \begin{equation}
        V = R \oplus I \oplus B \quad \text{with } \quad \omega = 
        \begin{pmatrix} \omega_R & 0 & 0 \\ 
                       0 & 0 &  \omega'' \\
                     0 & -\omega'' & 0  \end{pmatrix} ,
    \end{equation} 
    where $\omega_R$ is the induced symplectic form on $R\cong C/I$ and $\omega''$ is the natural pairing of $I$ and $B\cong V/C$. 
\end{prop}
\begin{proof}
   \begin{itemize}
    \item[]
    \item \textbf{Choice of $R$:} The projection $C \stackrel{\pi}{\to} C/I$ gives the classical coisotropic reduction from Definition \ref{def:coiso_red}. An arbitrary graded linear complement $R$ of $I\subset C$ is a symplectic subspace of $V$, as $(R, \restr{\omega}{R}) \cong (C/I, \omega_{C/I})$ and thus $\restr{\omega}{R}$ is nondegenerate. Moreover, by Lemma \ref{lemma:ortho_comp_squared}, $R^\omega $ is symplectic as well, since
        \begin{equation}
            R^\omega \cap R = 0, \quad \textnormal{so} \quad R^\omega \cap \left( R ^\omega \right) ^\omega = 0.
        \end{equation}       
       \item \textbf{Choice of $B$:} Using Lemma \ref{lemma:inductive_construction_of_B} of Appendix \ref{ssec:inductive_construction_of_B}, we construct $B$, an isotropic linear complement of $C$ satisfying        \begin{equation}\label{eq:condition_from_appendix}
        s^{-1}\dimsum{I}{s} = \dimsum{B}{s^{-1}}.
       \end{equation}
       To check that  $I$ and  $B$ are Lagrangian subspaces of $I \oplus B\cong R^\omega$, we check the condition from Lemma \ref{lemma:dimsum_symp},
    \begin{align*}
        s^{-1}\dimsum{I}{s} & \stackrel{?}{=} \dimsum{I \oplus B}{s^{-1}} - \dimsum{I}{s^{-1}} = \dimsum{B}{s^{-1}} ,\\
        s^{-1}  \dimsum{B}{s} &  \stackrel{?}{=} \dimsum{I \oplus B}{s^{-1}} - \dimsum{B}{s^{-1}} = \dimsum{I}{s^{-1}} .
    \end{align*} 
    These equalities follow from equation \eqref{eq:condition_from_appendix}. \qedhere

\end{itemize}
\end{proof}

\subsection{Linear and Lagrangian Relations}
Recall that a linear relation between vector spaces $U$ and $V$ is a linear subspace $L\subset U\times V$. This defines the category $ \mathcat{LinRel} $ of real finite-dimensional vector spaces and linear relations. Identity is given by the diagonal $\diag{-}$ and composition as the set-theoretic composition,
\begin{equation}\label{eq:comp}
    L_2 \after L_1 \define \{ (u, w)\in U\times W \mid \exists v \in V \text{ such that } (u, v)\in L_1 \text{ and } (v, w)\in L_2 \}.
\end{equation}
The image and kernel of a linear relation $L \subset U \times V$ are defined by
    \begin{align}
        \Im L & \define \left\{ v \in V \ \vert \ \exists u \in U \colon \left(u,v\right) \in L \right\}, \\
        \ker L & \define \left\{ u \in U \ \vert \  \left( u,0 \right) \in L \right\}.
    \end{align}
A transpose of a linear relation $L\subset U\times V$ is 
\begin{equation} L^T \define \{ (v, u)\in V\times U \mid (u, v) \in L \subset U\times V \}. \end{equation}

\noindent
$L$ is called \textbf{injective} if $\Ker{L} = 0$, \textbf{surjective} if $\Im{L}=V$, \textbf{coinjective} if $\Ker{L^T}=0$ and \textbf{cosurjective} if $\Im{L^T}=U$. The relation $L$ is a graph of a linear map $U \to V$ if and only if $L$ is cosurjective and coinjective. If this condition is not satisfied, one should view $L$ as a \emph{partially defined}, \emph{multi-valued} map; the \emph{domain of definition} is $\Im{L^T}$, and the \emph{indeterminacy} is $\Ker{L^T}$ (see e.g.\ \cite{mac_lane:additive_relations}).
\medskip

For a $(-1)$-shifted symplectic vector space $\left(V , \omega \right)$, define $\flip{V}$ as the same graded vector space with an opposite symplectic form $-\omega$. 

\begin{definition}\label{def:osc}
    The objects of the \textbf{linear $(-1)$-symplectic category} $\LinOSC$ are $(-1)$-shifted symplectic vector spaces and morphisms from $V$ to $W$ are \textbf{Lagrangian relations}, i.e.\ Lagrangian subspaces of $\flip{V} \times W$. The identity morphism is given by the diagonal $\diag{V} \subseteq \flip{V} \times V$. Composition of \begin{equation}\begin{tikzcd}
U \arrow[r, "L_1"] & V \arrow[r, "L_2"] & W
\end{tikzcd}\end{equation} is defined as a composition of relations of sets from equation \eqref{eq:comp}.
 \end{definition}

\begin{example}
    The basic example of a Lagrangian relation is the graph $\graph{\phi} \subseteq \flip{V} \times W$ of a symplectic isomorphism $\phi \colon V \to W$, In fact, all isomorpisms in $\LinOSC$ are obtained as such graphs, and the functor $\graph{(-)}$ identifies the category of symplectomorphisms with the maximal subgroupoid (the core) of $\LinOSC$. We will often denote these isomorphisms by a decorated arrow (see also Definition \ref{def:redcored})
    \vspace{-2mm}
    \begin{equation}
         \begin{tikzcd}
            V \arrow[r, "\graph{\phi}", two heads, tail] & W .
            \end{tikzcd}
    \end{equation}
\end{example}

\begin{lemma}\label{lemma:lagrangian_comp}
	The category $\LinOSC$ is well-defined. 
\end{lemma}	
\begin{proof}
The diagonal relation is Lagrangian and satisfies the identity axiom. Composition of set-theoretic relations is associative, and composing two linear relations gives again a linear relation.

\smallskip
To check that the composition of two Lagrangian relations is Lagrangian, we can use Lemma \ref{lemma:reduced_lagrangian} as in Weinstein \cite[Lecture 3]{weinstein:lectures}.  The subspace $C = \flip{U} \times \diag{V} \times W \subset \flip{U}\times V \times \flip{V}\times W$ is coisotropic, and its coisotropic reduction is $\flip{U} \times \diag{V} \times W/ (* \times \diag{V} \times *) \cong \flip{U}\times W$. The image $ \left[  L_1 \times L_2  \right]_C$ of the Lagrangian $L_1 \times L_2  \subset \flip{U}\times V \times \flip{V}\times W $ is $L_2 \circ 
 L_1$, which is therefore Lagrangian by Lemma~\ref{lemma:reduced_lagrangian}.
\end{proof}

\begin{remark}[Dagger compact closed category]
    The transpose $L \mapsto L^T$ defines a dagger on the symmetric monoidal category $(\OSC, \times)$ \cite[Def~2.2]{selinger:dagger_compact}. Moreover, with $\flip{V}$ as the dual object of $V$, $\OSC$ is a compact closed category, i.e.\ the  internal hom $[V_1, V_2]$ can be computed as $\flip{V_1} \times V_2$ \cite{kelly:compact_closed}. Finally, these are compatible as in \cite[Def.~2.6]{selinger:dagger_compact}, i.e.\ $\OSC$ is dagger compact closed. Dagger compact closed categories (originally introduced as strongly compact closed categories) are a natural setting for (finite-dimensional) quantum mechanics, as proposed by Abramsky and Coecke \cite{abramsky_coecke:1, abramsky_coecke:2}.
\end{remark}

\begin{example}[{Odd version of \cite[Thm.~4.8.1]{gs:semi-classical}}]\label{ex:shiftedtangent}
There is a \emph{shifted cotangent functor}
\begin{equation}
\sT  \colon \GrVect \to \LinOSC,
\end{equation}
defined on objects by $V\mapsto ( \sT V , \omega_\mathrm{can} )$ and on morphisms by sending $f\colon V \to W$ to the Lagrangian
	\begin{equation*}
        \sT  f = \{ (\beta\circ f, v, \beta, f(v)) \mid v \in V, \beta \in W^* \} \ \subset \ \s{V^*} \oplus V \oplus \s{W^*} \oplus W =\overline{\sT  V} \oplus \sT  W.
    \end{equation*}
    \vspace{-4mm}
\end{example}

We end this section by showing that each Lagrangian relation has a coisotropic image, with the corresponding isotrope being the kernel of the transposed relation (see e.g.\ \cite[{p.~945}]{gs:integralgeometry}).

\begin{lemma}\label{lemma:ker_L_T}
       Let $L \colon \left( U , \omega_U \right) \to \left( V , \omega_V \right)$ be a Lagrangian relation. Then \[\ker L^T = \left( \Im L \right)^{\omega_V} \quad \text{and} \quad \ \ker L = \left( \Im L^T \right)^{\omega_U}.\]
    \end{lemma}
    \noindent
    In particular, for a Lagrangian relation, surjectivity is equivalent to coinjectivity and injectivity is equivalent to cosurjectivity.
    \begin{proof}
        By definition, $v \in \ker L^T$ if and only if  $\left( 0,v \right) \in L = L^\omega$. Equivalently, for all $( u' , v' ) \in  L$,
        \begin{equation}
                  0 = -\omega_{U} \oplus \omega_{V} \left( \left( 0,v \right) , \left( u' , v' \right) \right)
        = -\omega_{U} \left( 0,u' \right) + \omega_{V} \left( v,v' \right) = \omega_{V} \left(v,v' \right) .
        \end{equation} Therefore, $v \in \ker L^T$ if and only if $v \in \left( \Im L \right)^{\omega_V}$.
        The second equation is proven from the first by considering $L^T$ in place of $L$.
\end{proof}

\subsection{Reductions and Coreductions}

\begin{definition}\label{def:redcored}
    Let $L \in U \to V$ be a Lagrangian relation. We say $L$ is 
    \begin{itemize}
        \item  a \textbf{reduction}, if $ \ker L^T = 0$ (equivalently $ \Im L = V$) and we denote $U \dhxrightarrow{L} V$, 
        \item a \textbf{coreduction}, if $ \ker L = 0$ (equivalently $ \Im L^T = V$) and we denote \begin{tikzcd}[cramped, column sep=0.5cm]
            U \arrow[r, "L", tail] & V
            \end{tikzcd}. \qedhere
    \end{itemize} 
\end{definition}

\noindent Some useful properties follow from the definition: 
\begin{itemize}
    \item A composition of (co)reductions is a (co)reduction.
    \item A Lagrangian relation is both a reduction and a coreduction \begin{tikzcd}[cramped, column sep=0.7cm]
        U \arrow[r, "L", two heads, tail] & V
        \end{tikzcd} if and only if it is an isomorphism in $\LinOSC$ (i.e.\ a graph of a symplectic isomorphism).
    \item A Lagrangian relation $L\colon U\to V$ is a reduction if and only if $L \after L^T = \id_V$, and a coreduction if and only if $L^T \after L = \id_U$. Moreover, reductions are epimorphisms, and coreductions are monomorphisms in $\LinOSC$ (see Remark \ref{rmk:epi}). We will denote the subcategory of reductions (epimorphisms) by $\LinRed$.
\end{itemize}

     The following proposition shows that every reduction is equal to coisotropic reduction $\text{red}_C$ from Definition \ref{def:coiso_red} up to a post-composition by a symplectic isomorphism
     \begin{equation*}
     \begin{tikzcd}
	V & R
	\arrow["L", two heads, from=1-1, to=1-2]	
\end{tikzcd}
     = \left(\begin{tikzcd}
	 V & & {\Im L^T/ \ker L} & R
	\arrow["\text{red}_{\Im \! L^T}",two heads, from=1-1, to=1-3]
	\arrow[ two heads, tail, from=1-3, to=1-4]
\end{tikzcd}\right).\end{equation*}
     This is a straightforward modification of \cite[Prop. 3.4.2]{gs:semi-classical}.

\pagebreak
\begin{prop}[Reductions are coisotropic reductions]\label{prop:reductions_are_coiso}
    Let $V,R \in \LinOSC$, $L \subseteq \flip{V} \times R$ a graded subspace. Then the following two conditions are equivalent.
    \begin{enumerate}
        \item \label{item:isred}$ L \in \LinOSC (V,R)$ and it is a reduction.
        \item \label{item:excoiso} There exists $ C \subseteq V $ coisotropic and a symplectic isomorphism $\phi \colon C/C^\omega \xrightarrow{\cong} R$ such that \begin{equation}
            L = \left\{ \left( c,r \right) \in \flip{V} \times R \ \vert \ c \in C,  r = \phi \left( \pi (c) \right) \right\}.
        \end{equation}
        where $\pi \colon C \to C/C^\omega$ is the quotient map.
    \end{enumerate}
\end{prop}
\begin{proof}
    First, we suppose that condition \ref{item:excoiso} holds. The relation $L$ is surjective by definition; let us prove that $L$ is indeed a Lagrangian relation. Denote $\pi_R \equiv \phi \after \pi$. Since $\phi$ is a symplectomorphism, we have $\restr{\omega}{C} = \pi_R^* \omega_R$. From this, the isotropy of $L$ follows:
    \begin{equation}
        -\omega \oplus \omega_R \left( \left( c , \pi_R c \right) , \left( c' , \pi_R c' \right) \right) = -  \omega \left( c , c' \right) + \pi_R^* \omega_R \left( c , c' \right) = 0 , \quad \text{for all } c, c' \in C.
    \end{equation}
    To show that $L$ is coisotropic, let us take arbitrary $(v, \pi_R(d))\in L^\omega$ with $v\in V, d\in C$. This element satisifes, for any $c \in C$,
    \begin{equation}
     0 = -\omega(v, c) + \omega_R(\pi_R(d), \pi_R(c)) = \omega(d-v, c).
    \end{equation}
    In other words, $v-d \in C^\omega$ and since $d\in C$, then $v\in C$ as well and $\pi(v) = \pi(d)$, which means that $L^\omega \subset L$.

\smallskip
    Now let condition \ref{item:isred} hold and $V \dhxrightarrow{L} R$ be a reduction, denote $C \equiv \Im L^T$. By Lemma \ref{lemma:ker_L_T}, $C$ is coisotropic: \begin{equation}
        (\Im L^T )^\omega = \ker L \subseteq \Im L^T.
    \end{equation}
    By coinjectivity of $L$, there exists a map $\pi_R \colon C \to R$ such that $L = \left\{ \left( c , \pi_R (c) \right) \in \flip{V} \times R \ \vert \ c \in C \right\}$.  The projection $\pi_R$ uniquely factors through $\pi$, since $\ker \pi = C^\omega = \ker L = \ker \pi_R$;  the induced symplectic isomorphism $\phi \colon C/\ker{\pi_R} \to \Im{\pi_R} $ 
    \begin{equation}\label{eq:reductionnotation}
        \begin{tikzcd}
            C \arrow[rd, "\pi_R"] \arrow[d, "\pi"'] &   \\
            \frac{C}{C^\omega} \arrow[r, "\phi"']        & R 
            \end{tikzcd}  
    \end{equation}
    is uniquely determined by $\phi(\pi(c)) = \pi_R(c)$.
\end{proof}
In other words, given a reduction $V \dhxrightarrow{L} R$, the coisotrope $C$ is unique since it is  determined by $\Im L^T$, and the symplectic isomorphism $\phi \colon  C/C^\omega \cong R$ is unique, as $\graph{\phi}$ is necessarily equal to the composition \begin{equation}\label{eq:isofromreduction}\graph{\phi} = \bigg(\begin{tikzcd}
	{C/C^\omega} & V & R
	\arrow["\text{red}_C^T",tail, from=1-1, to=1-2]
	\arrow["L", two heads, from=1-2, to=1-3]
\end{tikzcd}\bigg),\end{equation}
since $\graph{\phi}\circ\operatorname{red}_C = L$ implies $\graph{\phi}\circ\operatorname{red}_C\circ\operatorname{red}_C^T = L \circ\operatorname{red}_C^T$ and $\operatorname{red}_C\circ\operatorname{red}_C^T = \id$ for a reduction $\operatorname{red}_C$.

\subsubsection{Factorization}

Crucially, it turns out that any Lagrangian relation can be factored into a reduction followed by a coreduction, with coisotropics given by $\Im L^T$ and $\Im L$ respectively. This is a $(-1)$-shifted symplectic version of the usual claim for linear relations \cite[p.~1045]{mac_lane:additive_relations} or Lagrangian relations \cite[p.~946]{gs:integralgeometry}.

\begin{prop}\label{prop:red_cored_factorization}
    Let $L \in \LinOSC (U,V)$. Let  $L_U$ and $L_V$ be the coisotropic reductions with respect to $\Im L^T \subset U$ and $\Im L \subset V$, respectively. Then  $L_V \circ L \circ L_U^T=:\graph{\phi}$ is an isomorphism and the following diagram commutes.
\begin{equation}\label{eq:factorization}
\begin{tikzcd}
    U \arrow[rd, "L_U"', two heads] \arrow[rrr, "L"] &                                          &                        & V \\
                                                          & \Im L^T/\Ker L \arrow[r, "\graph{\phi}"', two heads, tail] & \Im L/\Ker L^T \arrow[ru, "L_V^T"', tail] &  
    \end{tikzcd}
\end{equation}

Moreover, this factorization is unique in the following sense. For every factorization $L = L_2^T \circ L_1$ where $L_1\colon U \dhxrightarrow{} R$, $L_2 \colon V \dhxrightarrow{} R$ are reductions, there are unique isomorphisms $\psi_1$, $\psi_2$ making the following diagram commute.
\begin{equation}\label{diag:uniqueness_of_factorization}
\begin{tikzcd}
	U && V \\
	& R \\
	{\Im L^T /\Ker L} && {\Im L / \Ker L^T}
	\arrow["L", from=1-1, to=1-3]
	\arrow["{L_1}"{description}, from=1-1, to=2-2, two heads]
	\arrow["{L_2^T}"{description}, from=2-2, to=1-3, tail]
	\arrow["{L_U}"', from=1-1, to=3-1, two heads]
	\arrow["{L_V^T}"', from=3-3, to=1-3, tail]
	\arrow["{\exists!\graph{\psi_1}}"{description}, dashed, from=3-1, to=2-2]
	\arrow["{\exists! \graph{\psi_2}}"{description}, dashed, from=3-3, to=2-2]
	\arrow["{\graph{\phi}}"' , from=3-1, to=3-3, two heads, tail]
\end{tikzcd}
\end{equation}
\end{prop}
\begin{proof}
    By Lemma \ref{lemma:ker_L_T}, $\Im L^T$ is coisotropic in $U$ and $\Im L$ in $V$. Thus, the quotients $R_U \define \Im L^T /\Ker L$ and $R_V \define \Im L /\Ker L^T$ are $(-1)$-shifted symplectic and define reductions $L_U$ and $L_V$ by Proposition \ref{prop:reductions_are_coiso}.   Denote the quotient maps by $\pi_U \colon \Im L^T \to R_U$, $\pi_V \colon \Im L \to R_V$ and the composition by $\Phi := L_V \circ L \circ L_U^T \subset \flip{R_U} \times R_V $. In this notation, 
    \begin{equation}
        \Phi = \left\{ \left( r ,s \right) \in \flip{R_U} \times R_V \ \vert \ \exists  (u,v) \in L \colon r=\pi_U (u) , s =   \pi_V (v) \right\}.
    \end{equation}
     To show that $\Phi$ is an isomorphism, it suffices to notice that $\Phi$ is both surjective and cosurjective, and hence $\Phi = \graph{\phi}$ for a symplectic isomorphism $\phi \colon R_U \to R_V$. To check that diagram \eqref{eq:factorization} commutes, it is easy to see that $L \subset L_V^T \circ \Phi \circ L_U$, which implies equality of these two Lagrangian subspaces of $\flip{U}\times V$.
     \smallskip     

    Turning to \eqref{diag:uniqueness_of_factorization}, for any such factorization $L = L_2^T\circ L_1$, we have $\Im{\left(L_2^T\circ L_1\right)} = \Im{L_2^T}$, since $L_1$ is surjective. Thus, the coisotrope giving the reduction $L_2$ is necessarily equal to $\Im{L}$, and similarly $\Im{L_1^T} = \Im{L^T}$. By Proposition \ref{prop:reductions_are_coiso}, we get unique $\psi_{1,2}$ making the left and right triangles in \eqref{diag:uniqueness_of_factorization} commute. The bottom triangle commutes since the whole square commutes and we have \eqref{eq:isofromreduction}.
\end{proof}

\begin{example}
    \label{ex:continuingcotangentexample} Continuing Example \ref{ex:shiftedtangent}, we can interpret Proposition \ref{prop:red_cored_factorization} for linear maps. If $L = \sT f$ for $f\colon U \to V$, we get
    \begin{eqnarray*}
        \Ker{L} =  0 \oplus \Ker{f} \;\subset\; \s{U^*} \oplus U, \quad & \quad \Im{L} = V^*\oplus \Im{f} \; \subset \; \s{V^*} \oplus V, \\
        \Ker{L^T} =  \Ker{f^\text{t}}\oplus 0 \subset  \s{V^*} \oplus V, & \Im{L^T} = \Im{f^t} \oplus U \subset \s{U^*} \oplus U.
    \end{eqnarray*}
    Then, Lemma \ref{lemma:ker_L_T} says that $\Ker{f^\text{t}} = \Ann{\Im{f}}$, while Proposition \ref{prop:red_cored_factorization} gives the isomorphism $U/\Ker{f} \cong \Im{f}$.
\end{example}

\begin{remark}[Epimorphisms are reductions] \label{rmk:epi}
We can now show that not only a reduction $L$ satisfies $L \after L^T = \id$ and is therefore an epimorphism, but the other implication is also true. We can decompose any epimorphism $L \in \LinOSC (U,V)$ as $L = L_V^T \after L_U = L_V^T \after L_V  \after L_V^T \after L_U $ where $L_U$ and $L_V$ are reductions. Since $L$ is epic, we have $L_V^T \after L_V = \id_V$ and thus $L_V$ is an isomorphism and $L$ is a reduction.
\end{remark}
     
\begin{definition}\label{def:factor_cospan}
     Given a Lagrangian relation $L \colon U \to V$, we define its \textbf{factorization cospan} to be a pair of reductions $(L_U , L_V)$ 
    \begin{equation}\label{diag:factor_cospan}
        \begin{tikzcd}
U \arrow[rd, "L_U"', two heads] \arrow[rr, "L", dashed] &   & V \arrow[ld, "L_V", two heads] \\
                        & R &                        
\end{tikzcd}
    \end{equation}
    such that
    \[
        L=L_V^T \after L_U. \qedhere
    \]
    \end{definition}

    Proposition \ref{prop:red_cored_factorization} shows that a factorization cospan always exists and moreover, for a fixed relation $L$, all factorization cospans are uniquely isomorphic in the sense of diagram \eqref{diag:uniqueness_of_factorization}; we will speak of \emph{the} factorization cospan of a Lagrangian relation. Note that since $L_V \after L_V^T = \id_V$, the diagram \eqref{diag:factor_cospan} is commutative in $\OSC$.

\subsection{Spans of Reductions}
\label{sec:double_reductions}
If we have two Lagrangian relations $U \xrightarrow{L_1} V \xrightarrow{L_2} W$, we can form their factorization cospans and get the following diagram in $\OSC$.
\begin{equation}\label{diag:twocompred}\begin{tikzcd}
	U && V && W \\
	& R && {\tilde{R}}
	\arrow["{L_1}", from=1-1, to=1-3]
	\arrow["{L_2}", from=1-3, to=1-5]
	\arrow[two heads, from=1-1, to=2-2]
	\arrow["L"', two heads, from=1-3, to=2-2]
	\arrow["{\tilde{L}}", two heads, from=1-3, to=2-4]
	\arrow[two heads, from=1-5, to=2-4]
\end{tikzcd}\end{equation}
Let us now investigate two natural questions connected with diagram \eqref{diag:twocompred}: whether we can complete the \emph{span} of reductions $R \dhxleftarrow{} V \dhxrightarrow{} \smash{\tilde{R}}$ to a commutative square, and how this square relates to the usual definition of a composition of cospans in terms of pushouts \cite[Sec.~I.2.6]{BenabouBicategories}. We will answer these questions in Corollary \ref{thm:cospans_comp}.
\smallskip 

Let us therefore consider an arbitrary \emph{span of reductions}, i.e.\ is a pair of reductions $ \smash{(L , \tilde{L})}$ with a~common source.
    \begin{equation}\label{eq:doublereduction}
        \begin{tikzcd}
            & V \arrow[ld, "L"', two heads] \arrow[rd, "\tilde{L}", two heads] &           \\
          R &                                                          & \tilde{R}
          \end{tikzcd}
        \end{equation}
\noindent
Transposing $L$, we get a relation $\tilde{L} \circ L^T\colon R \to \tilde{R}$, which can be easily described as
\begin{equation}\label{eq:redcoredexplicit}
    \tilde{L} \circ L^T = \{ (\pi_R(c), \pi_{\tilde{R}}(c)) \mid c \in C \cap \tilde{C} \} \subset R \times \tilde{R},
\end{equation}
where $C\define \Im{L^T}$ and $\pi_R \colon C \to R$ are the coisotrope and projections such that $L = \{(\pi_{R}(c),c)\mid c \in C\}$ using the notation from diagram \eqref{eq:reductionnotation} (and analogously for $\tilde{C}$ and $\pi_{\tilde{R}}$).

\subsubsection{Orthogonal Spans of Reductions}\label{ssec:ortho_double_red}

There is a natural choice of the completion of diagram \eqref{eq:doublereduction} to a square, namely the factorization cospan $R \dhxrightarrow{} S_0 \dhxleftarrow{} \smash{\tilde{R}}$ of the composite $\smash{\tilde{L} \circ L^T\colon R \to \tilde{R}}$. However, the resulting square of reductions is not always commutative.\footnote{As a counter-example, consider a decomposition $V=R \oplus I \oplus B$ from Proposition \ref{prop:red_decomp} let $R$ and $\tilde{R}$ be reductions along $R\oplus I$ and $R\oplus B$, respectively. Then $\tilde{L} \after L^T = \id_R$ and $S_0 \cong R$, but $L \neq \tilde{L}$ unless $I = B = *$.} We will now completely characterize the class of spans of reductions for which this happens.

\begin{definition}\label{def:ortho}
    We say a span of reductions $R \dhxleftarrow{L} V \dhxrightarrow{\tilde L} \tilde{R}$ is \textbf{orthogonal} if $\Ker{L} \perp \Ker{\tilde{L}} $, i.e.\ if for all $i \in \Ker{L}$, $\tilde{i} \in \Ker\tilde{L}$,
    \[
        \omega (i , \tilde{i}) = 0 . \qedhere
    \]
\end{definition}
\noindent Denote $I\define \Ker {L}$ and $\tilde{I}\define \Ker{\tilde{L}}$. Observe that $I \perp \tilde{I}$ is equivalent to $I+\tilde{I}$ being isotropic, which is in turn equivalent to  $C \cap \tilde{C}$ being coisotropic, since $C \cap \tilde{C} = ( I + \tilde{I} )^\omega$. 

\begin{theorem}\label{thm:ortho}
    Consider a span of reductions $R \dhxleftarrow{L} V \dhxrightarrow{\tilde L} \tilde R $. Then the factorization cospan 
\[\begin{tikzcd}
	R && {\tilde{R}} \\
	& {S_0}
	\arrow["{K_0}"', two heads, from=1-1, to=2-2]
	\arrow["{\tilde{L} \circ L^T}", from=1-1, to=1-3]
	\arrow["{\tilde{K}_0}", two heads, from=1-3, to=2-2]
\end{tikzcd}\]
makes the square of reductions \eqref{diag:fact_cospan_square}  commute if and only if $(L, \tilde L)$ is an orthogonal span of reductions. 
\begin{equation}\label{diag:fact_cospan_square}
 \begin{tikzcd}
	& V \\
	R && {\tilde R} \\
	& S_0
	\arrow["L"', two heads, from=1-2, to=2-1]
	\arrow["{\tilde L}", two heads, from=1-2, to=2-3]
	\arrow["{K_0}"', two heads, from=2-1, to=3-2]
	\arrow["{\tilde{K}_0}", two heads, from=2-3, to=3-2]
\end{tikzcd}   
\end{equation}
    Moreover, any cone of reductions under an orthogonal span of reductions $(L, \tilde L)$ uniquely factors through the cone \eqref{diag:fact_cospan_square}. In other words, the pushout of \eqref{eq:doublereduction} in the category of reductions exists if and only if the span of reductions is orthogonal, and it is given by the factorization cospan of $\smash{\tilde L \circ L^T}$.
\end{theorem}
A cone of reductions $\tilde R \dhxrightarrow{} S \dhxleftarrow{} \tilde R$  under $(L, \tilde L)$ is equivalently given by a reduction  $M \colon V \dhxrightarrow{} S$ which factors through both $L$ and $\smash{\tilde{L}}$. We will thus first study the problem of factoring one reduction through another; Theorem \ref{thm:ortho} will follow by applying the following proposition twice.
    \begin{prop}\label{prop:factorizingreductions}
Consider a pair of reductions $L\colon V \dhxrightarrow{} R$ and $M \colon V\dhxrightarrow{} S$.
\[\begin{tikzcd}
	& V \\
	R \\
	& S
	\arrow["L"', two heads, from=1-2, to=2-1]
	\arrow["M", two heads, from=1-2, to=3-2]
	\arrow["{\exists ? K}"', dashed, two heads, from=2-1, to=3-2]
\end{tikzcd}\]
Then the following are equivalent:
\begin{enumerate}[label=(\arabic*)]
    \item \label{i:pfr:factors} $M$ factors through $L$, i.e.\ there is a reduction $K \colon R \dhxrightarrow{} S$ such that $K \circ L = M$,
    \item \label{i:pfr:eqonS} $M \circ L^T \circ L = M$,
    \item \label{i:pfr:eqonV} $M^T \circ M \circ L^T \circ L = M^T \circ M$,
    \item \label{i:pfr:coisotropes} $\Im M^T \subseteq \Im L^T$.
\end{enumerate}
Moreover, if any of these conditions holds, one has $K = M \circ L^T$ and thus $K$ is unique if it exists.
\end{prop}
\begin{proof} 
    If $K$ exists, then $K \circ L = M$ implies that  $K = M \circ L^T$ and thus $K$ is unique. Moreover, any relation $K$ making the diagram above commute is necessarily a reduction, since $K \circ L = M$ is surjective. 
    
    Let us now show the equivalence of the four statements. The equivalence of \ref{i:pfr:factors} and \ref{i:pfr:eqonS} is clear since we just need to check that $K \circ L = M$ for our only candidate $K = M \circ L^T$. The equivalence of \ref{i:pfr:eqonS} and \ref{i:pfr:eqonV} is due to $M \after M^T = \id$, as $M$ is a reduction. Finally, let $C_L \define \Im L^T$ and $C_M \define \Im M^T$, then $M^T \circ M$ is the relation\footnote{In fact, the assignemt $C_M \mapsto M^T \circ M$ gives a bijection between coisotropic subspaces $M \subset V$ and symmetric idempotent endomorphisms of $V$ in $\LinOSC$ \cite[Eq.~9.14]{gs:integralgeometry}. The present proposition can be understood as saying that this bijection is order-preserving, with respect to the partial order on idempotents from \cite{Mitsch1986}. See also \cite[{\textsection29}]{HalmosITHS} for an analogous statement for Hilbert spaces.} 
    \[ M^T \circ M = \{ (c, c') \in C_M\times C_M \mid c - c' \in (C_M)^\omega \} \subseteq V\times V. \]
    Similarly, the composition $M^T \circ M \circ L^T \circ L$ is equal to 
    \[ M^T \circ M \circ L^T \circ L = \{ (c + i_L, c+ i_M) \mid c \in C_L \cap C_M, i_L \in (C_L)^\omega, i_M \in (C_M)^\omega \}  \subseteq V\times V.\]
    The statement \ref{i:pfr:eqonV} is equivalent to the inclusion $M^T \circ M \circ L^T \circ L \subseteq M^T \circ M$, as both are Lagrangian subspaces of $\flip{V}\times V$. This is in turn equivalent to the following three conditions for all $c \in  C_r \cap C_M$, $i_L \in (C_L)^\omega$, $i_M \in (C_M)^\omega$:
    \vspace{-2mm}
    \begin{align*}
        c + i_L & \in C_M, \\
        c+ i_M & \in C_M, \\
        i_L - i_M &\in (C_M)^\omega.
    \end{align*}
    The last condition is equivalent to $(C_L)^\omega \subseteq (C_M)^\omega$, i.e.\ $C_M \subseteq C_L$, and this implies the first condition as $c+i_L \in C_M\cap C_L + (C_L)^\omega \subseteq C_M \cap C_L + (C_M)^\omega \subseteq C_M$. The middle condition is always satisfied. Thus, the inclusion $M^T \circ M \circ L^T \circ L \subset M^T \circ M$ is equivalent to $C_M \subseteq C_L$, i.e.\ the statement \ref{i:pfr:coisotropes}. 
\end{proof}

\begin{proof}[Proof of Theorem \ref{thm:ortho}]
    Consider an arbitrary cone of reductions, and denote $M= K \circ L = \tilde K \circ \tilde L$.
    \begin{equation}\label{diag:square_of_red}
\begin{tikzcd}
	& V \\
	R && {\tilde R} \\
	& S
	\arrow["L"', two heads, from=1-2, to=2-1]
	\arrow["{\tilde L}", two heads, from=1-2, to=2-3]
	\arrow["K"', two heads, from=2-1, to=3-2]
	\arrow["{\tilde{K}}", two heads, from=2-3, to=3-2]
	\arrow["M"{description}, two heads, from=1-2, to=3-2]
\end{tikzcd}
\end{equation}
The reduction $M$ is given (up to a unique isomorphism of $S$) by the coisotrope $D \define \Im M^T$. From Proposition \ref{prop:factorizingreductions} we see that necessarily $D \subset \Im L^T \cap \Im  \tilde L^T$, and vice versa choosing coisotropic $D\subset \Im L^T \cap \Im \tilde{L}^T$ gives $K$ and $\tilde K$ such that $K \circ L = M = \tilde K \circ \tilde L$ (again from Proposition \ref{prop:factorizingreductions}). Thus, commutative squares of the form \eqref{diag:square_of_red} (up to an isomorphism in $S$) are in bijection with coisotropic subspaces of $\Im L^T \cap \Im \tilde{L}^T$. In particular, such commutative squares exist if and only if $\Im L^T \cap \Im \tilde{L}^T$ is itself coisotropic,\footnote{If a subspace $E\subset V$ contains a coisotropic subspace $C \subset E$, then $E$ is also coisotropic, as $E^\omega \subset C^\omega \subset C \subset E$.} i.e.\ if $(L, \tilde L)$ is an orthogonal span of reductions.
\smallskip

Next, we show that the case $\smash{D_0 = \Im {L^T} \cap \Im {\tilde{L}^T}}$ corresponds to $\smash{(K_0, \tilde{K}_0)}$ being the factorization cospan of $\smash{\tilde L \circ L^T}$. By uniqueness of factorization from Proposition \ref{prop:red_cored_factorization}, it is enough to check that $\smash{\tilde{L} \circ L = \tilde{K}_0^T \circ K_0}$. The relation $\smash{\tilde{L} \circ L}$ is given by \eqref{eq:redcoredexplicit},
\begin{equation*}
    \tilde{L} \circ L^T = \{ (\pi_R(c), \pi_{\tilde{R}}(c)) \mid c \in C \cap \tilde{C} \} \subset R \times \tilde{R},
\end{equation*}
\noindent
while $\tilde{K}_0^T \circ K_0$ is equal to
\[ \tilde{K}_0^T \circ K_0   = \tilde L \circ M_0^T \circ M_0 \circ L^T = \{ (\pi_{R}({c_1}), \pi_{\tilde{R}}({c_2})) \mid c_1, c_2 \in C \cap \tilde C, c_1 - c_2 \in (C\cap \tilde C)^\omega\} \in R\times \tilde{R}. \]
By choosing $c_1 = c_2$, we see that $\tilde{L} \circ L^T \subset \tilde{K}_0^T \circ K_0$, which proves the equality of these Lagrangian subspaces of $\flip{R}\times \tilde{R}$.

\smallskip
To prove the pushout property of this square, consider an arbitrary commutative square of reductions (the outer square on the diagram).
\[\begin{tikzcd}[row sep = 1cm]
	& V \\
	R & {S_0} & {\tilde R} \\
	& S
	\arrow["L"', two heads, from=1-2, to=2-1]
	\arrow["{\tilde L}", two heads, from=1-2, to=2-3]
	\arrow["K"', two heads, from=2-1, to=3-2]
	\arrow["{\tilde{K}}", two heads, from=2-3, to=3-2]
	\arrow["{K_0}", from=2-1, to=2-2]
	\arrow["{\tilde {K}_0}"', from=2-3, to=2-2]
	\arrow["N"{description,pos=0.35}, dashed, two heads, from=2-2, to=3-2]
	\arrow["{M_0}"{description}, two heads, from=1-2, to=2-2]
\end{tikzcd}\]
To show $S_0$ is a pushout, we need to construct $N$ as above and check that it is a map of cocones. The reduction $N$ is constructed by factorizing $\smash{M = K\circ L = \tilde K \circ \tilde L}$ through $M_0$, which exists and is unique by Proposition \ref{prop:factorizingreductions}. Finally, we need to check $K = N \circ K_0$, which is equivalent to $K \circ L = N \circ K_0 \circ L$, i.e.\ $M = N \circ M_0$, and similarly for $\tilde K = N \circ \tilde{K}_0$.
\end{proof}

\subsection{Category of Cospans of Reductions}\label{sec:comp_and_factor}
Recall from Section \ref{sec:double_reductions} that we wanted to investigate the factorization cospan of a composition. Let us consider diagram \eqref{diag:twocompred} and add the factorization of $\tilde{L} \circ L^T$ to the bottom.\footnote{We drop the subscript ${}_0$, which denoted the factorization cospan in the previous section, to lighten the notation.}

\begin{equation}
    \label{eq:compfactcospans}
\begin{tikzcd}
	U && V && W \\
	& R && {\tilde{R}} \\
	&& S
	\arrow["{L_1}", from=1-1, to=1-3]
	\arrow["{L_2}", from=1-3, to=1-5]
	\arrow[two heads, from=1-1, to=2-2]
	\arrow["{L^T}", tail, from=2-2, to=1-3]
	\arrow["{\tilde{L}}", two heads, from=1-3, to=2-4]
	\arrow[tail, from=2-4, to=1-5]
	\arrow["K"', two heads, from=2-2, to=3-3]
	\arrow["{\tilde{K}^T}"', tail, from=3-3, to=2-4]
\end{tikzcd}\end{equation}
Since the square and the two triangles commute, the outer triangle gives the factorization cospan of $L_2\circ L_1$, by uniqueness of factorizations from Proposition \ref{prop:red_cored_factorization}.
This way, we have defined the composition of morphisms in the following category.
\begin{definition}
    The \textbf{category of cospans of reductions}, denoted $\Cospan$, has $(-1)$-symplectic vector spaces as objects and isomorphism classes of cospans of reductions as morphisms. Composition of the factorization cospans of $L_1 : U \to V$ and $L_2 :V \to W$ is given by the factorization cospan of $L_2 \after L_1$.
\end{definition}
Usually, the composition in the category of cospans is defined using pushouts. Theorem \ref{thm:ortho} says these two compositions agree exactly when the pushout is defined, i.e.\ when $\smash{(L, \tilde{L})}$ is an orthogonal span of reductions. We now summarize these results in the following corollary.

\begin{cor}[Factorization cospan as an equivalence of categories]\label{thm:cospans_comp} There is an equivalence of categories between the linear $(-1)$-shifted symplectic category and the category of isomorphism classes of cospans of reductions in the linear $(-1)$-shifted symplectic category given by the construction of the factorization cospan;
\begin{equation}
    \OSC \cong \Cospan.
\end{equation}
\noindent
Moreover, assume the Lagrangian relations $L_1$ and $L_2$ \textbf{compose orthogonally}, i.e.\ $\ker L_1 ^T \perp \ker L_2$. Then the composition in $\Cospan$ coincides with the composition of factorization cospans under $L_1$ and $L_2$ along the pushout in the category of reductions $\LinRed$, as displayed in diagram \eqref{eq:comp_cospan_red}.
    \begin{equation}\label{eq:comp_cospan_red}
\begin{tikzcd}
U \arrow[rr, "L_1", dashed] \arrow[rd,  two heads] &                               & V \arrow[rr, "L_2", dashed] \arrow[ld, "L"', two heads] \arrow[rd, "\tilde{L}", two heads] &                                                & W \arrow[ld, two heads]\\
                                                                                                       & R \arrow[rd, "K"', two heads] &                                                                                            & \tilde{R} \arrow[ld, "\tilde{K}", two heads] &                                                                                     \\
                                                                                                       &                               & R'  \arrow[uu, phantom, "\lrcorner"{anchor=center, pos=0.125, rotate=135}]
                                                                                        &                                                &                                                                                    
\end{tikzcd}
\end{equation}    
\end{cor}
\begin{proof}
    The first part follows from Proposition \ref{prop:red_cored_factorization} and the construction above. For the second part, $L_1$ and $L_2$ compose orthogonally if and only if $(L,\tilde{L})$ is an orthogonal span of reductions. This assumption enables us to transpose any coreductions in the diagram \eqref{eq:compfactcospans} into reductions while preserving its commutativity, using Theorem \ref{thm:ortho}. 
\end{proof}

\section{Half-Densities and Perturbative BV Integration}
In the previous section, the degree of the symplectic form did not play a big role, apart from some complications when calculating dimensions. In this chapter, we will introduce notions for which it is essential that the symplectic form has an odd degree: half-densities and their perturbative Batalin-Vilkovisky integrals. Our goal is to define a fiber Batalin-Vilkovisky integral along a reduction.
\subsection{Linear Half-Densities}

Densities on graded vector spaces\footnote{Recall that we now assume that our vector spaces are finite-dimensional.} are real-valued functions on the set of bases, transforming with the Berezinian. In general, Berezinian is defined for even automorphisms of free modules over a commutative superalgebra \cite[Sec.~3.3]{ManinGFT}. We will restrict to the case of linear automorphisms of graded vector spaces, where the usual formula simplifies as follows. We will also replace the field $\mathbb R$ with the field $\Rhbar$ of formal Laurent series in powers of $\hbar$ to simplify $\hbar$-dependent calculations. Unless it is explicitly denoted otherwise, the tensor product $\otimes$ is understood over $\Rhbar$.

\begin{definition}
    Let $ A$ be an invertible degree-preserving linear map on a graded\footnote{Graded in $\Z$, which induces a $\Z_2$-grading by parity.} vector space $V = V \even \oplus V \odd$, which naturally decomposes into $A = A \even \oplus A\odd$. We define the \textbf{Berezinian} of $A$ as
\[
    \Ber{ A} = \frac{ \det  A\even }{ \det  A\odd }. \qedhere
\]
\end{definition} 
\noindent We list some simple properties:
\begin{itemize}
    \item  $\Ber{A_1 A_2}= \Ber{A_1}\Ber{A_2}$, $\Ber{ A^{-1}} = \Ber{ A}^{-1} $,
    \item $\Ber{ A^T}=\Ber{ A}$,
    \item  $\Ber{ A_1 \oplus  A_2}= \Ber{ A_1}\Ber{ A_2}$,
    \item For $V$ a $(-1)$-shifted symplectic space and $ A \colon V\to V$ a symplectic isomorphism, $\Ber{ A} = \left( \det  A\even \right)^2 $. 
    
\end{itemize}

\begin{definition} Let $\base{e}$ be a basis of a graded vector space $V$. A \textbf{linear density of weight} $\alpha \in \Real$ on $V$ is a map $\rho \colon \base{e} \mapsto \rho \left( \base{e} \right) \in \Rhbar $ satisfying
    \begin{equation}
        \rho \left( \base{e} \cdot  A \right) = |\Ber{ A}|^\alpha \rho \left( \base{e} \right)
    \end{equation}
    for any invertible linear map $ A$ of degree 0, which acts naturally on $\base{e}$ from the right.
    We denote the one-dimensional vector space of linear densities of weight $\alpha$ on $V$ by $\DensWeight{\alpha}{V}$. For $V=*$ a point, we define\footnote{Motivated by the fact that a zero-dimensional vector space has a unique basis, or by Lemma \ref{lemma:short_exact_dens} with $W=0$.} $\DensWeight{\alpha}{*} = \Rhbar$. We call elements of $\smash{\HalfDens{V}}$ \textbf{linear half-densities}.\footnote{On $(-1)$-shifted symplectic vector spaces, they are the natural objects to integrate along Lagrangian subspaces, see Section \ref{ssec:weightsandBV}.}
\end{definition}

 There is a natural notion of multiplication of densities: $(\rho \cdot \sigma) ( \base{e} ) \define \rho ( \base{e}) \sigma ( \base{e}) $. We can use this to identify a priori different spaces of linear densities.

\begin{lemma}\label{lemma:density_properties} There are following canonical isomorphisms (which we will denote by the \say{$=$} sign).
    \begin{equation}
        \DensWeight{\alpha}{V} \otimes \DensWeight{\beta}{V} = \DensWeight{\alpha+\beta}{V} , \quad \DensWeight{\alpha}{V}=\DensWeight{-\alpha}{V^*}, \quad \DensWeight{\alpha}{V} = \DensWeight{-\alpha}{\s V}.
    \end{equation}
\end{lemma}
\begin{proof}
    The first isomorphism is the multiplication of densities. The second isomorphism sends a density $\rho\in \DensWeight{\alpha}{V}$ to the density $\tilde{\rho}(\base{e}^*) \define \rho(\base{e})$, where $\base{e}^*\subset V^*$ is the dual basis to $\base{e}\subset V$. Transforming $\base{e}$ by $ A$ transforms $\base{e}^*$ by $( A^{-1})^T$, hence $\tilde\rho$ has weight $-\alpha$. The last isomorphism follows form the fact that exchanging $ A\even \leftrightarrow  A\odd$ inverts the Berezinian. 
\end{proof} 

\begin{lemma}\label{lemma:short_exact_dens}
   A short exact sequences of graded vector spaces of the form
    \begin{equation}
        \begin{tikzcd}
            0 \arrow[r] & U \arrow[r , "i"] & V \arrow[r, "p"] & W \arrow[r] & 0
            \end{tikzcd}
    \end{equation}
    induces a canonical isomorphism $ \DensWeight{\alpha}{V} =  \DensWeight{\alpha}{U} \otimes \DensWeight{\alpha}{W}$.
\end{lemma}
\begin{proof}
    Analogously to the classical case \cite[]{gs:semi-classical}. Different extensions of a basis of $i(U)$ to $V$ differ only by an action of $ A$ with upper triangular block matrix structure on $U \oplus W$. Since the block structure is induced on both $U\even \oplus W\even$ and $U\odd \oplus W\odd$, the Berezinian does not depend on the choice of such extension. Then, by $\Ber{ A' \oplus  A}=\Ber{ A'}\Ber{ A}$, the Lemma follows.
\end{proof}
\noindent
Since in $\GrVect$, $U \oplus W \cong U \times W$, a simple corollary of Lemma \ref{lemma:short_exact_dens} is
\begin{equation}
    \DensWeight{\alpha}{U \times W} = \DensWeight{\alpha}{U} \otimes \DensWeight{\alpha}{W}.
\end{equation}

\subsubsection{Linear Half-densities and Measures on Kernels of Lagrangian Relations} \label{ssec:weightsandBV}

    One reason for introducing half-densities is that they induce densities on Lagrangian subspaces. Indeed, for a Lagrangian subspace $L\subset V$, there is a following sequence of canonical isomorphisms (due to \cite[Eq.~(30)]{schwarz:geometry_of_bv}, \cite[Eq.~(3.5)]{khudaverdian} and \cite[Sec.~3]{khudaverdianvoronovDFOSG}).
    \begin{align}  \HalfDens{V} 
     \cong \HalfDens{L} \otimes \HalfDens{V/L} 
    \cong \HalfDens{L} \otimes \HalfDens{\s{L^*}} 
     \cong \HalfDens{L} \otimes \HalfDens{L}
    \cong \DensWeight{}{L},\end{align}
    where we used Lemmata \ref{lemma:density_properties} and \ref{lemma:short_exact_dens}. The isomorphism $V/L \cong \s{L^*}$ is given by $[v]\mapsto \omega(v, -)$. The appearence of the shift $\s{L^*}$ highlights the difference between even and odd symplectic geometry.\footnote{For classical even $\omega$, we get $\DensWeight{\alpha}{V}  
    = \DensWeight{\alpha}{L} \otimes \DensWeight{\alpha}{ L^*} 
    = \DensWeight{\alpha}{L} \otimes \DensWeight{-\alpha}{L} = \R$, that is a proof of the existence of canonical symplectic volume.} This argument can be generalized to a more general reduction $V \dhxrightarrow{} R$ instead of $L\colon V \dhxrightarrow{} *$.
    \begin{lemma}\label{lemma:weightsandBV}
        Let $I\subset V$ be an isotropic subspace, $C \define I^\omega$ and $R \define C/I$. Then there is a canonical isomorphism
        $\HalfDens{V} = \HalfDens{R} \otimes \Dens{I}$.
    \end{lemma}
    \begin{proof} Using Lemmata \ref{lemma:density_properties}, \ref{lemma:short_exact_dens} and the isomorphism $ V/C \cong \s{I^*}$, we have
        \begin{align*}
    \HalfDens{V}
       & \cong  \HalfDens{C} \otimes \HalfDens{V/C},
    \\ & \cong  \HalfDens{R} \otimes \HalfDens{I} \otimes \HalfDens{V/C},
    \\ & \cong  \HalfDens{R} \otimes \HalfDens{I} \otimes \HalfDens{\s{I^*}},
    \\ & \cong  \HalfDens{R} \otimes \HalfDens{I} \otimes \HalfDens{I},
    \\ & \cong  \HalfDens{R} \otimes \DensWeight{}{I}. \qedhere
\end{align*}
    \end{proof}

\subsection{Formal Functions}

A central part of the Batalin-Vilkovisky formalism are integrals of functions of the form $e^{S/ \hbar}$. In order to accomodate such functions and their products, we will consider formal polynomials in $V^*$ and $\hbar^{\pm 1}$, completed with respect to the weight grading of \cite[Sec.~2.2]{braun_maunder:minimal_models}, see \cite[Sec.~2.2]{doubek_jurco_pulmann:quantum_L_infty_and_HPL} for an analysis that easily translates to our setting.
\medskip

The \textbf{weight} of a homogeneous element $$f \in  \Sym^k(V^*) \otimes_\R \R \hbar^g \subset \widehat{\mathrm{Sym}}(V^*) [[ \hbar , \hbar^{-1} ]]$$
is defined to be $w = 2g + k$. Here, $\Sym^k$ is the graded-symmetric algebra given by the quotient of $V^{\otimes k}$ by the ideal generated by elements of the form $\alpha \otimes \beta - (-1)^{\deg{\alpha}\deg{\beta}} \beta \otimes \alpha $ and $\Sym^0 (V) \equiv \Real$ by definition. $\smash{\widehat{\mathrm{Sym}}}$ denotes the completion of the space of polynomial elements to formal series.

\begin{definition}\label{def:formal_functions}
Let $\Fw{V}$ be the space of finite linear combinations of homogeneous weight $w$ elements
\[\Fw{V} \define \bigoplus_{\substack{ k \geq 0 \\ g \in \Z \\ 2g+k=w }} \Sym^k(V^*) \otimes_\R \R \hbar^g.\]
We define the \textbf{space of formal functions} on a graded vector space $V$ as the space of formal series with weight bounded from below
\[
    \F{V} \define \Big\{ f \in  \prod_{w \in \Z} \Fw{V} \mid \text{ the weight components of $f$ vanish for $w < N_f$ for some $N_f \in \mathbb Z$} \Big\}. \qedhere
\]
\end{definition}
\noindent
This way, $\F{V}$ is an algebra, and the product preserves the weight grading. We would like to emphasize that the weight grading just solves a technical triviality.

\subsubsection{BV Algebra}
We will now briefly recall the Batalin-Vilkovisky structure \cite{bv} on the space $\F{V}$, see \cite[Sec.~2.1]{doubek_jurco_pulmann:quantum_L_infty_and_HPL} for more details.  Let $\{ e_i \} _i$ be a basis of $V \in \LinOSC$, $\{ \phi^i \}_i$ the dual basis. We define the matrix $\omega_{ij} \define \omega (e_i , e_j)$ and denote its inverse $\omega^{ij}$. The algebra $\F{V}$ is then spanned by graded-commutative polynomials in $\phi^i$ and $\hbar^{\pm 1}$.

\begin{definition}\label{def:bv}
Define the \textbf{odd Poisson bracket} $\{-,-\} \colon \F{V} \times \F{V} \to \F{V} $ by\footnote{Right partial derivatives are defined by
\begin{equation*}
    \rpartial{F}{\phi^i} = (-1)^{\deg{i} (\deg{F} - \deg{i})} \lpartial{F}{\phi^i}.
\end{equation*}}
    \begin{equation}
        \left\{ f , g \right\} \define \rpartial{f}{\phi^i} \omega^{ij} \lpartial{g}{\phi^j},
    \end{equation}  
     and the \textbf{BV Laplacian} $\BV \colon \F{V} \to \F{V} $ by
      \begin{equation}\label{eq:bv_laplacian}
        \BV  \define \frac12 (-1)^{\deg{i}} \omega^{ij} \frac{\partial^2_L}{\partial \phi^i \partial \phi^j}.
    \end{equation}
    Both of these maps are defined to be $\Rhbar$-linear.
\end{definition}
These two operations define a BV algebra structure on $\F{V}$, i.e.\ $\BV$ is a second-order differential operator of degree 1 which squares to $0$, and $\{-,-\}$ is a degree 1 Poisson bracket satisfying
\begin{equation}
            \BV (fg) =  (\BV f) g + \sign{\deg{f}} f \BV g + \sign{\deg{f}} \{ f , g \}.
\end{equation}
    \begin{remark} Using the odd Poisson bracket, we can give yet another equivalent formulation of orthogonality of spans reductions from Definition \ref{def:ortho}. Let $C, \tilde C \subset V$ be two coisotropic subspaces, with corresponding isotropes $I, \tilde{I}$. Define the \textbf{vanishing ideal} $\I_C$ of $C$ by $\I_C \define \langle \Ann{C}\rangle \subseteq \F{V}$. Then
    \begin{equation}
        I\perp \tilde{I} \quad \text{if and only if} \quad \{ \I_C , \I_{\tilde{C}} \} \subseteq \I_C + \I_{\tilde{C}}.
    \end{equation}
    \end{remark}

Finally, we extend the BV Laplacian to the space of all half-densities, to be thought of as the sections of the square root of the Berezinian bundle over $V$.
\begin{definition}
    The space of \textbf{half-densities} $\D{V}$ on a graded vector space $V$ is defined as the tensor product 
    \begin{equation}
        \D{V} \define \F{V} \otimes \HalfDens{V}.
    \end{equation}
    The BV Laplacian $\BV \colon \D{V} \to \D{V}$ is defined by $\BV \otimes \id\colon \F{V} \otimes \HalfDens{V} \to \F{V} \otimes \HalfDens{V}$.
\end{definition}
\begin{remark}\label{rem:funczerodim}
Note that our definitions imply that for $V = *$, the algebra $\F{V}$ is equal to the algebra of formal Laurent series $\mathbb{R}((\hbar))$, and similarly $\D{V} \cong \mathbb{R}((\hbar))$.
\end{remark}

In fact, it is the space of half-densities on an odd symplectic supermanifold which carries a canonical BV Laplacian \cite{khudaverdian}, see also \cite{severa:origin_of_bv}. In our case, when the manifold is the vector space $V$, there is a unique-up-to-rescaling translation-invariant half-density which induces the BV operator on functions from Definition~\ref{def:bv}.

 \subsection{\texorpdfstring{$(-1)$}{(-1)}-Shifted Symplectic dg Vector Spaces}
 We will equip some of the symplectic vector spaces with a compatible differential. However, we would like to point out that we will not consider symplectic dg vector spaces as objects of a symplectic category; it will be the morphisms which will carry the differential (see Definition \ref{def:gen_Lagr}).
 \begin{definition} \label{def:dg_symp} A \textbf{$(-1)$-shifted symplectic dg vector space} is a $(-1)$-shifted symplectic space $(V, \omega)$ equipped with a degree $1$ differential $Q\colon V \to V$ such that for all $v,w \in V$,
 	\[
 	\omega \left( Q v , w \right) + (-1)^{\deg{v}} \omega \left( v , Q w \right) = 0. \qedhere
 	\]
 \end{definition}
Such differentials are in bijection with elements $\Sfree \in \Sym^2(V^*)$ of degree 0 such that
\begin{equation*}
    \left\{ \Sfree , \Sfree \right\} = 0,
\end{equation*}
i.e.\ $\Sfree$ satisfies the \emph{classical master equation}. This bijection is given by\footnote{If we denote
	$Q(e_j) = Q_j^i e_i,$ $\Sfree = \frac 12 s_{ij}\phi^i \phi^j $
    then the above formula gives 
	$Q^i_j = - \omega^{ik} s_{kj}$. Equivalently, we have $\Sfree(v, w) = (-1)^{\deg{v}} \omega(Qv, w)$.}
\begin{equation}\{\Sfree, - \} \define Q^\text t
    \end{equation}
    where the transpose of a graded linear map $f$ is defined on $\phi \in V^*$ by $f^t ( \phi ) \define (-1)^{(\deg{f}+1)\deg{\phi}} \phi \after f$. We will thus use both $Q$ and $\Sfree$ to refer to a differential on a $(-1)$-shifted symplectic vector space.

 \subsubsection{Non-Degenerate Reductions and Canonical Decomposition}\label{ssec:canon_decomp}
 We will now study ``non-degenerate'' reductions, i.e.\ reductions along which we can define a perturbative Gaussian integral with the kernel given by $e^{\Sfree/\hbar}$. In some forms and special cases, this non-degeneracy condition is well-known among experts and appears e.g.\ in \cite[Lemma 2.5.1, Section 2.7]{costello}. As explained in Proposition  \ref{prop:SDRnondeg} below, such reductions succinctly encode special deformation retracts of symplectic vector spaces or abstract Hodge decompositions \cite{chuang_lazarev:hodge_decomposition} in the language of Lagrangian relations.

 \begin{definition}\label{def:nondeg}
 	Let $\left( V , \omega , Q \right)$ be a $(-1)$-shifted symplectic dg vector space and $\Sfree \in \Sym^2(V^*)$ the corresponding quadratic form. An isotrope $I \subset V$ is said to be \textbf{non-degenerate} if it satisfies any of the following equivalent conditions.
 	\begin{itemize}
 		\item  $\Sfr{I}\define \restr{\Sfree}{I}$, the restriction of $\Sfree$ to $I$, is a non-degenerate pairing.
            \item The matrix $\omega ( e_i , Q e_j )$ is non-degenerate for any basis $\{ e_i \}$ of $I$.
 		\item $I \cap (QI)^{\omega} = \{0\}.$
 	\end{itemize}
    A reduction $L: V \dhxrightarrow{} R$ is said to be a \textbf{non-degenerate reduction} if $\ker L$ is non-degenerate.
 \end{definition}

The utility of this definition is demonstrated in Proposition \ref{prop:canonicaldecomposition}, which proves that any non-degenerate isotrope $I \subset V$ determines a \textbf{canonical decomposition} in the sense of Proposition \ref{prop:red_decomp},
\begin{equation}\label{eq:canonicaldecomposition}
    V = (I \oplus QI)^\omega \oplus I \oplus QI \eqqcolon R_\can \oplus I \oplus B.
\end{equation}

 \begin{prop}\label{prop:canonicaldecomposition}
 	Let $I \subset V$ be a non-degenerate isotrope with respect to $Q$. Then:
 	\begin{enumerate}
        \item $I \cap \Ker Q = \{0\}$.
 		\item $I \cap QI = \{0\}$.
        \item $QI$ is isotropic.
 		\item   $I \oplus QI$ is symplectic.
 		\item \label{item:Iomega} $ ( I \oplus QI )^\omega \oplus I = I^\omega$, thus also $I^\omega / I \cong ( I \oplus QI )^\omega \eqqcolon R_\can $. 
        \item \label{item:decQ} In the decomposition $V  = I \oplus ( I \oplus QI )^\omega \oplus QI$, the only non-zero components of the differential $Q$  are $\restr{Q}{I} \colon I \xrightarrow{\cong} QI$ and possibly $\restr{Q}{R_\can} \colon ( I \oplus QI )^\omega \to ( I \oplus QI )^\omega$.
        
 	\end{enumerate}
 \end{prop}
 \begin{proof}
 	\begin{enumerate}
        \item[]
        \item If there were a vector $i\in I \cap \Ker{Q}$, then $\Sfree(i, -) = \pm \omega(Qi, -) = 0$ and thus $\Sfr{I}$ would be degenerate.
 		\item From $Q^2=0$, we have $I \cap QI \subset \ker Q \cap I$, which is zero by the previous point.
        \item As $Q$ is compatible with $\omega$, we get $\omega \left(Qi,Qi'\right) = \pm \omega 
 		\left(i,Q^2i'\right) = 0$ for all $ i, i' \in I$.
 		\item By $\Sfree = \pm \omega(Q-, -)$, we have that $ \restr{\omega}{ I \oplus QI}$ is block-diagonal with $\pm \Sfr{I}$ on anti-diagonals. Since this matrix is invertible, $I \oplus QI$ is symplectic (i.e.\ it does not intersect its $\omega$-orthogonal complement). 
 		\item We have $I \subset  I \oplus QI$, and so $(I \oplus QI)^\omega \subset I^\omega$ and also $I \subset I^\omega$. Together, this gives $(I \oplus QI)^\omega \oplus I \subset I^\omega$. The equality is proven by comparing dimensions. Since $Q \colon L \to QL$ is an isomorphism of degree 1, we have $\dimsum{QL}{s} = s\dimsum{L}{s}$ Using Lemma \ref{lemma:dimsum_symp}, we get that
 		\begin{equation} \dimsum{(QI \oplus I)^\omega\oplus I}{s} = s(\dimsum{V}{s^{-1}} - \dimsum{QI\oplus I}{s^{-1}}) +\dimsum{I}{s}= s(\dimsum{V}{s^{-1}} - (1+s^{-1})\dimsum{I}{s^{-1}}) +\dimsum{I}{s}\end{equation}
 		
 		while
 		\begin{equation} \dimsum{L^\omega}{s} = s(\dimsum{V}{s^{-1}} - \dimsum{I}{s^{-1}}). \end{equation}
 		The difference is
 		\begin{equation}\dimsum{I}{s} - \dimsum{I}{s^{-1}} = \sum_{k \ge 1} (\dim I_k - \dim I_{-k})s^k, \end{equation}
 		which vanishes since $I$ has a non-degenerate pairing $\Sfr{I}$ of degree $0$.

        \item We have $Q = \restr{Q}{I} + \restr{Q}{QI} +  \restr{Q}{( I \oplus QI )^\omega}$. The first map is the isomorphism $I\to QI$. The second map vanishes. The third map lands again in $( I \oplus QI )^\omega$, since $\omega(Qr, i + Qi') = \pm \omega(r, Qi) = 0$. \qedhere
 	\end{enumerate}
 \end{proof}
\noindent
This decomposition of $V$ induces a special deformation retract between $V$ and $R$, see e.g.\ \cite{crainic:hpl}.
 \begin{definition}
    A \textbf{special deformation retract} is a pair of dg vector spaces, chain maps $i, p$ and a degree $(-1)$ map $k$ as below
    \vspace{-2em}
    \begin{equation}
    \begin{tikzcd}
    {\left( V , Q_V \right)} \arrow[rr, "p", shift left] \arrow["k"', loop, distance=3em, in=215, out=145] &  & {\left( W , Q_W \right)} \arrow[ll, "i", shift left]
    \end{tikzcd}
    \vspace{-1em}
    \end{equation}
    such that $pi = \id_W$, $ip = \id_V + Q_V k + k Q_V$, $k^2 = 0$, $pk = 0$ and $ki = 0$. If $V$ and $W$ are $(-1)$-shifted dg symplectic, then we say that the special deformation retract is \textbf{symplectic} if $i$ is a symplectic map, $p$ is a Poisson map and $k$ satisfies $\omega_V(kv, v') = (-1)^{\deg{v}} \omega_V(v, kv')$.
\end{definition}
Finally, we can relate symplectic SDRs, and also abstract Hodge decompositions of Chuang and Lazarev \cite[Def.~2.1]{chuang_lazarev:hodge_decomposition} with non-degenerate reductions.

\begin{prop} \label{prop:SDRnondeg}
There is a bijection between the following structures.
\begin{enumerate}
    \item Non-degenerate isotropes in $V$.
    \item Symplectic special deformation retracts between $V$ and some $R$ (up to an isomorphism of $R$).
    \item Abstract Hodge decompositions $s, t \colon V \to V$ of $(V, \omega, Q)$.
\end{enumerate}
The abstract Hodge decomposition is harmonious (loc.cit.), i.e.\ $R$ is isomorphic to the homology of $V$, if and only if $\restr{Q}{R_\can} = 0$.
\end{prop}
\noindent
See \cite[Proposition 2.5]{chuang_lazarev:hodge_decomposition} for a related statement in the harmonious case.
\begin{proof}
    Given a non-degenerate isotrope $I\subset V$, one can take define a SDR $V \rightleftarrows I^\omega/I \cong (I \oplus QI)^\omega $ using the projection $p \colon V \to  (I \oplus QI)^\omega$ and inclusion $i \colon (I \oplus QI)^\omega \to V$ from the decomposition \eqref{eq:canonicaldecomposition}; $k$ is equal to $-(\restr{Q}{I})^{-1} \colon QI \to I$. Conversely, a symplectic SDR defines  a non-degenerate isotrope $I\define \Im{k}$. Indeed, given $k(v) \in \Im k \cap (Q (\Im k) )^\omega$, for all $v' \in V$,
    \begin{align*}
        0 = & \ \omega (k(v), Qk(v')) = 
        \omega (k(v),ip(v')) -
        \omega (k(v),v') -
        \omega (k(v),kQ (v')) = \\
        = &
        \pm \omega (v,kip(v')) -
        \omega (k(v),v') 
        \pm \omega (v, k^2Q (v')) = \omega (k(v),v'),
   \end{align*}
   which forces $k(v)=0$ and thus $\Im k$ is non-degenerate.
    
    Similarly, an abstract Hodge decomposition is defined from a symplectic special deformation retract by setting $s\define k$ and $t\define ip$, and given $(s, t)$, one can take $I\define \Im{s}$.\smallskip

   It is then a straightforward check that these maps are well defined bijections.
\end{proof}
  
\subsection{Perturbative BV Integral}
In this section, we recall the notion of perturbative BV integrals. Our goal is to define a formal Laurent series
\begin{equation}
    \intBV{L\subset V}{} f \rho \in \Rhbar 
\end{equation}
where $L\subset V$ is a Lagrangian and $f \otimes \rho$ is a half-density on $V$. This integral is usually defined  choosing a Lagrangian complement to $L$ and computing the ordinary Berezin-Lebesgue integral over L \cite{schwarz:geometry_of_bv}. Alternatively, one can define the perturbative version of this integral using homological perturbation theory, see \cite{albert_bleile_frohlich:bv_integrals, gwilliam_jf:how_to_feynman_diagrams, gwilliam-thesis}.
\medskip

We choose a third approach and define this integral axiomatically, which quickly leads to explicit formulas. This is possible for non-degenerate $\smash{\restr{\Sfree}{L}}$, essentially because we can use the canonical decomposition from the previous section. The axiomatic approach makes it easy to compare to other approaches; see Remark \ref{rem:BerLeb} for a comparison with the Berezin-Lebesgue integral and Section \ref{ssec:fibintHPL} for an equivalence with homological perturbation theory. These two comparisons also allow for easy proofs of the uniqueness and existence of this integral.
\begin{definition}\label{def:BVintegralLagrangian}
	Let $(V, \omega , \Sfree )$ be a $(-1)$-shifted symplectic dg vector space and $L\subset V$ a Lagrangian subspace such that the  $\Sfr{L}$ is non-degenerate (Definition \ref{def:nondeg}). Then
	\begin{equation} \label{eq:integralLagr} \intBV{L\subset V}{L} (-) \colon \D{V} \to \mathbb R((\hbar)), \end{equation}
	denoted by ${\displaystyle f\otimes \rho \mapsto \intBV{L\subset V}{L} f\rho}$,
	is the $\Rhbar$-linear weight-homogenous\footnote{Our normalization of the integral implies that it has weight equal to $\operatorname{sdim}{L}= \dim L_\textnormal{even} - \dim L_\textnormal{odd}$} map uniquely specified by:
	\begin{enumerate}
		\item \label{i:dBVi:exact} ${\displaystyle\intBV{L\subset V}{L} (\{\Sfree, f\} + \hbar\BV f) \rho= 0}$.
		\item \label{i:dBVi:vanish} ${\displaystyle\intBV{L\subset V}{L}g\rho = 0}$ for any $g\in \I_L \subset \F{V}$, i.e.\ integration annihilates the vanishing ideal of $L$.
		\item  \label{i:halfdensnorm} ${\displaystyle\intBV{L\subset V}{L}  \rho = (2\pi)^{\tfrac {\dim L_\textnormal{even}}2} \hbar^{\tfrac {\dim L_\textnormal{even} - \dim L_\textnormal{odd}}2} \rho( \base{e}_L,  Q(\base{e}_L) )}$ for any basis $\base{e}_L$ of $L$. \qedhere
	\end{enumerate}
\end{definition}
\noindent
Note that the last condition is independent of $\base{e}_L$ as the RHS is equal to the composition $\HalfDens{V} = \HalfDens{L}\otimes \DensWeight{-\frac{1}{2}}{L} = \mathbb \Rhbar$, where we use $V\cong L \oplus QL \cong L \oplus \s L$ and Lemma \ref{lemma:density_properties}.
\medskip

To ensure that the integral from Definition \ref{def:BVintegralLagrangian} is uniquely defined, we will show in the next section that \eqref{eq:integralLagr} is given by the famous Wick's Lemma. To ensure that such an integral exists, one could check directly that the prescription given by Wick's Lemma satisfies the properties listed in Definition \ref{def:BVintegralLagrangian}. We instead use the fact that (up to normalization on linear half-densities specified by Item \ref{i:halfdensnorm}) this integral coincides with the perturbed projection constructed using homological perturbation lemma, see Proposition \ref{prop:fiberBVisHPL}.

\begin{remark}\label{rem:BerLeb}
    The first two items of Definition \ref{def:BVintegralLagrangian} are motivated by usual properties of BV integrals: the integral vanishes on $\BV$-exact half-densities\footnote{This is the odd Stokes' theorem of Schwarz \cite[Thm.~2]{schwarz:geometry_of_bv}.} and depends only on the restriction of $f$ to $L$.
    
    The third item fixes a normalization of the integral that matches the usual Gaussian integrals, see also \cite[Eq.~(49)]{schwarz:semiclassical_in_BV}. Indeed, for $V = \sT  \R^k$, $L=\R^k$ and $\rho=1$ for the canonical basis,
    \begin{equation} \intover{\R^k \subset \sT  \R^k} e^{\frac 12 s_{ij} x^i x^j/\hbar} \rho = (2\pi)^{-\tfrac k2} \rho( \base{e}_L, Q(\base{e}_L) ) = (2\pi\hbar)^{-\tfrac k2}  \left\lvert\Ber{
       \begin{matrix} 1 & 0 \\
        0 & \frac{s_{ij}}{\hbar} \end{matrix} 
    } \right\rvert^{\frac 12} = \left\lvert\det\frac{s_{ij}}{2\pi\hbar}\right \rvert^{-\frac{1}{2}}\,. \end{equation}
    \noindent
    Similarly, for $V = \sT (\shift{1}{\R^k}\oplus  
 \shift{-1}{\R^k})$, $L=  \shift{1}{\R^k}\oplus  
 \shift{-1}{\R^k}$ and $\rho=1$ for the canonical basis
    \begin{equation}
    \intover{\substack{\shift{1}{\R^k}\oplus  
 \shift{-1}{\R^k} \phantom{\sT  (  )} \\ 
    \cap \phantom{\sT  (  )}   \\  \sT (\shift{1}{\R^k}\oplus  
 \shift{-1}{\R^k})} } e^{w_{ij} \eta^i \xi^j/\hbar} \rho =  \hbar^{-k}\rho( \base{e}_L,  Q(\base{e}_L) ) = 
      \left\lvert\Ber{
       \begin{matrix} w_{ij} & 0 & 0 & 0 \\
        0 & -w_{ji} & 0 & 0 \\
        0 & 0 & 1 & 0 \\
        0 & 0 & 0 & 1
        \end{matrix} } \right\rvert^\frac 12 = \left\lvert\det{\frac{w_{ij}}{\hbar}}\right\rvert ,
        \end{equation} 
    where the first two rows/columns in the matrix correspond to the fiber (even) coordinates and the last two to base (odd) coordinates.

    \smallskip
    The first integral agress with the (even) Gaussian integral for $s_{ij}$ negative definite. The second integral recovers the Berezin integral up to a sign.\footnote{To get correct signs for such Gaussian integrals, we would need to discuss orientations, which is orthogonal to the goals of this work.}
\end{remark}

\subsubsection{Properties of the BV Integral}

To prove Wick's Lemma, we will use a version of the Schwinger-Dyson equation, for context see e.g.\ \cite[Eq. 15.25]{quantization_of_gauge_systems}.

\begin{lemma}[Schwinger-Dyson equation]\label{lemma:schwinger-dyson}
    Let $\beta \in \Ann{L} \subset \I _L, f \in \F{V}$. 
    \begin{equation}
        \intBV{L\subset V}{L}  \{ \Sfree, \beta \} f  \rho 
        =               
        -(-1)^{\deg{ \beta }}  \hbar \intBV{L\subset V}{L}  \{ \beta, f \}  \rho
    \end{equation}
\end{lemma}

\begin{proof}
    The axiom \ref{i:dBVi:exact} of Definition \ref{def:BVintegralLagrangian} gives us
    \begin{equation} \intBV{L\subset V}{L} (\{\Sfree, \beta f\} + \hbar\BV (\beta f)) \rho = 0 \end{equation}
    which can be expanded using the properties of $\BV$ and $\{-,-\}$. Three of the five terms vanish by axiom \ref{i:dBVi:vanish} and by the fact that $\BV \beta = 0$ since $\beta$ is of polynomial degree 1.
\end{proof}

\noindent Now we can relate integrals of homogeneous polynomials of degree $k+1$ with integrals of homogeneous polynomials of degree $k-1$ using the non-degenerate pairing $\Sfr{L}$, arriving at Wick's lemma. 

\begin{lemma}[Wick's Lemma]\label{lemma:wick} Let $ \left( V , \omega , \Sfree \right)$ be a $(-1)$-shifted symplectic dg vector space, $L \subset V$ Lagrangian such that $\Sfr{L}$ is non-degenerate. Choose a basis $\{ \gamma^i \}_i $ of $L^*$. Let $\Sfr{L} \equiv s_{ij} \gamma^i \gamma^j$ and denote $s^{ij}$ its inverse. Then for any $ k \geq 1$:
    \begin{equation}
        \intBV{L\subset V}{L} \gamma^{i_1} \dots \gamma^{i_{2k}} \rho = \sum_{\sigma \smallin \mathrm{Pair}(2k)} \sign{\sigma} \hbar^k \bigg( \prod_{\mathrm{(j_1,j_2) \smallin  }\sigma} -s^{j_1 j_2} \bigg) \intBV{L\subset V}{L}  \rho  
    \end{equation}
    where $\mathrm{Pair}(2k)$ is the set of  $(2k-1)!! $ partitions of the set $\{ 1 , \ldots , 2k \} $ into disjoint pairs. The sign $(-1)^\sigma$ is obtained by bringing each variable $\gamma^{j_1}$ to the immediate left of its partner $\gamma^{j_2}$ assigned by the pairing $\sigma$.
\end{lemma}
\noindent
The proof of this lemma is a standard calculation. Let us finish this section by proving a version of Fubini's theorem.
\begin{prop}[Fubini's Theorem]\label{prop:Fubini}
    For $i= 1, 2$, let $(V_i, Q_i)$ be $(-1)$-shifted symplectic dg vector spaces and let $L_i \subset V_i$ be non-degenerate Lagrangian subspaces. Then
     \begin{equation}  \intBV{L_1\subset V_1}{1}(-) \otimes \intBV{L_2\subset V_2}{2}(-) = \intover{L_1\times L_2 \subset V_1 \times V_2} e^{(\Sfr{1}+\Sfr{2})/\hbar}(-) 
     \end{equation}
    as maps from $\D{(V_1 \times V_2)} \cong \D{V_1}\otimes \D{V_2}$.
\end{prop}
\begin{proof}
    It is easy to see that the LHS satisfies the three conditions from Definition \ref{def:BVintegralLagrangian}. 
    \begin{enumerate}
        \item The RHS is defined using the induced $(-1)$-shifted symplectic structure $\{ , \}$, $\BV$ on $V_1 \times V_2$ and $\Sfree = \Sfr{1} + \Sfr{2}$.  Using the isomorphism $\F{(V_1\times V_2)} \cong \F{V_1} \otimes \F{V_2}$, the operator $\hbar \BV + \{\Sfree, -\}$
        becomes 
        \begin{equation}
            (\hbar \BV_1 + \{\Sfr{1}, -\}_1) \otimes \id + \id \otimes (\hbar \BV_2 + \{\Sfr{2}, -\}_2),
        \end{equation}
        which is annihilated by $\intBV{L_1\subset V_1}{1}(-) \otimes \intBV{L_2\subset V_2}{2}(-) $.
        \item The vanishing ideal $\I_{L_1 \times L_2}$ is generated by elements of $\Ann{L_1} \times 0 $ or $0\times \Ann{L_2}$; both cases are annihilated by $\intBV{L_1\subset V_1}{1}(-) \otimes \intBV{L_2\subset V_2}{2}(-) $.
        \item We take $\rho = \rho_1 \otimes \rho_2$, and compute
        \begin{align*}
            \intBV{L_1\subset V_1}{1}(\rho_1) \otimes\intBV{L_2\subset V_2}{2}(\rho_2) 
        & = (2\pi)^{\dots} \hbar^{\dots} \rho_1(\base{e}_{L_1} , Q_1(\base{e}_{L_1})) \cdot \rho_2(\base{e}_{L_2} , Q_2(\base{e}_{L_2})) \\
        & = (2\pi)^{\dots} \hbar^{\dots} \rho_1\otimes \rho_2(\base{e}_{L_1} \sqcup  \base{e}_{L_2},  (Q_1+Q_2)(\base{e}_{L_1} \sqcup \base{e}_{L_2})  ).
        \end{align*}
        Here, the normalizations $(2\pi)^{\dots} \hbar^{\dots}$ match as the exponents are additive. \qedhere
    \end{enumerate}
\end{proof}

\subsubsection{Fiber Integrals along Non-degenerate Reductions}\label{ssec:fiberintegralsreductions}

Using the BV integral along a Lagrangian subspace (Definition \ref{def:BVintegralLagrangian}) and the canonical decomposition \eqref{eq:canonicaldecomposition} of the source of a non-degenerate reduction from Proposition \ref{prop:canonicaldecomposition}, we can now define (fiber) \textbf{integrals along reductions}. 

\begin{remark}
    Using Lemma \ref{lemma:weightsandBV}, we can motivate the following construction as follows, without any non-canonical choices of decompositions of $V$. Let 
    \begin{equation}
        f\otimes \rho \in \D{V} =  \F{V} \otimes \HalfDens{V}
         = \F{V} \otimes \HalfDens{R} \otimes \Dens{I}
    \end{equation}
    and decompose $\rho = \rho_R \otimes \rho_I \in \HalfDens{R} \otimes \Dens{I}$. Now we can integrate $f \! \! \mid_C \ \in \F{C}$ along $I$ to get a function $f_R \in \F{R} = \F{(C/I)}$ using the translation-invariant measure $\rho_I$. We are left with $f_R \otimes \rho_R \in \D{R} $.
\end{remark}

\begin{definition}
	Let $(V, \omega, \Sfree)$ be a $(-1)$-shifted symplectic dg vector space and $L:V \dhxrightarrow{} R$ a non-degenerate reduction, denote $\ker L = I$. Consider $V = R_\can \oplus I \oplus QI$ the canonical decomposition \eqref{eq:canonicaldecomposition} from Proposition \ref{prop:canonicaldecomposition}. We define 
	\begin{equation}
 \intBV{L}{I} \colon \D{V} \to \D{R} 
 \end{equation}
	by the composition \vspace{-1.5em}
	\begin{equation}
 \D{V} \cong \D{(I \oplus QI)} \otimes \D{R_\can} \xrightarrow{\intBV{I \subset I \oplus QI}{I} \otimes \id_{R_\can}} \D{R_\can} \cong \D{R},
 \end{equation}
    where $\Sfr{I}$ is induced on $\F{(I \oplus QI)}$ by the restriction of the quadratic function $\Sfree$ to $I$.
\end{definition}

The decomposition from Proposition \ref{prop:canonicaldecomposition} implies that $\Sfree = \Sfr{I} + \Sfr{R_\can}$. The first term is used for the integral, while the second term induces a canonical \textbf{transferred differential} $Q^R$ on $R$ (see also Appendix \ref{sec:transfer_of_Q}).

\begin{prop}
\label{prop:fiber_integral}
    This integral satisfies the following three axioms.
    \begin{enumerate}
        \item \label{item:fiberintBVcomm}
        ${\displaystyle \intBV{L}{I} \circ (\{ \Sfree, - \} + \hbar \BV) = (\{ \SfreeR, - \}_R + \hbar \BV_R)  \circ \intBV{L}{I}}$. 
        \item ${\displaystyle\intBV{L}{I}g\rho = 0}$ for $g\in \I_C$, i.e.\ integration annihilates the vanishing ideal of $C\equiv I^\omega$.
	\item  ${\displaystyle\intBV{L}{I}  \rho \in \HalfDens{R}}$ satisfies ${\displaystyle (\intBV{L}{I}  \rho) (\base{e}_R) = (2\pi)^{\tfrac{\dim I_\textnormal{even}}2} \hbar^{\tfrac{\dim I_\textnormal{even}-\dim I_\textnormal{odd}}2}  \rho(\base{e}_R, \base{e}_I, Q \base{e}_I)}$ \\ for any basis $\base{e}_I$ of $I$.
    \end{enumerate}
\end{prop}
\begin{proof} 
    \begin{enumerate}
    \item[]
    \item 
    The odd Poisson bracket and the BV Laplacian split between the two symplectic subspaces $V = R \oplus R^\omega$, as follows from Proposition \ref{prop:red_decomp}. By part \ref{item:decQ} of Proposition \ref{prop:canonicaldecomposition}, the only non-zero components of the differential $Q$  are $\restr{Q}{I} \colon I \rightarrow QI$ and $\restr{Q}{R} \colon R \to R$. So after $\restr{(\{ \Sfree, - \} + \hbar \BV)}{I \oplus QI} $ is annihilated by the axiom \ref{i:dBVi:exact} of Definition \ref{def:BVintegralLagrangian}, all that is left is $(\{ \SfreeR, - \}_R + \hbar \BV_R)$, which commutes with the integral.
        \item From part \ref{item:Iomega} of Proposition \ref{prop:canonicaldecomposition}., $C \equiv I^\omega = ( I \oplus QI )^\omega \oplus I $ and $\I_C = \left\langle \Ann{C} \right\rangle = \left\langle {(QI)}^* \right\rangle $, and we use Item \ref{i:dBVi:vanish} of Definition \ref{def:BVintegralLagrangian}.
    \item Let us choose $\rho_R \in \HalfDens{R}$ arbitrary. By Lemmata \ref{lemma:short_exact_dens}, \ref{lemma:weightsandBV}, there exists a unique half-density $\rho_{R^\omega}$ such that $\rho = \rho_R \otimes \rho_{R^\omega}$. Then 
    \begin{align*}
        (\intBV{L}{I}  \rho) (\base{e}_R) &= (\rho_R \otimes \intBV{I\subset I \oplus QI}{I}  \rho_{R^\omega}) (\base{e}_R) \\
        &=(2\pi)^{\tfrac{\dim I_\textnormal{even}}2} \hbar^{\tfrac{\dim I_\textnormal{even}-\dim I_\textnormal{odd}}2}   \rho_R (\base{e}_R) \rho_{R^\omega}  (\base{e}_I, Q \base{e}_I) 
        \\ &=  (2\pi)^{\tfrac{\dim I_\textnormal{even}}2} \hbar^{\tfrac{\dim I_\textnormal{even}-\dim I_\textnormal{odd}}2}   \rho(\base{e}_R, \base{e}_I, Q \base{e}_I). \qedhere
    \end{align*} 
\end{enumerate}
\end{proof}

\begin{lemma}\label{lemma:fubinitransfer} Let $V \dhxrightarrow{L} R \dhxrightarrow{L'} R'$ be non-degenerate reductions with respect to $Q$ and the transferred differential $Q_R$ respectively. Then the composition $V\dhxrightarrow{} R'$ is again non-degenerate and
    \begin{equation}
    \intBV{L'}{I'} \circ \intBV{L}{I} = \intBV{L' \after L}{I\oplus I' \! \! }.    
    \end{equation}
\end{lemma}

\begin{proof}
    Using the canonical decomposition twice, we get
    \begin{equation*}
        V \cong  I \oplus R  \oplus QI \cong  I \oplus ( I' \oplus R' \oplus Q I') \oplus QI
    \end{equation*}
    and the decomposition of the differential (Proposition \ref{prop:canonicaldecomposition}, Item \ref{item:decQ}) implies $\Sfree = \Sfr{I} + \Sfr{I'} +\Sfr{R'}$. The composition of transfers along $L$ and $L'$ is given by 
    \begin{align*} 
    \D{(I \oplus QI)} \otimes \D{(I' \oplus QI')} \otimes \D{R'}   &\xrightarrow{ \intBV{I\subset I\oplus QI}{I} \otimes \id_{I \oplus QI'} \otimes\id_{R'}} \D{(I' \oplus QI')} \otimes \D{R'}  \\ &\xrightarrow{\intBV{I'\subset I'\oplus QI'}{I'} \otimes \id_{R'}}  \D{R'},  
    \end{align*}
    i.e.\ $\intBV{I\subset I\oplus QI}{I} \otimes \intBV{I'\subset I'\oplus QI'}{I'} \otimes \id_{R'}$. This is equal to $ \intover{\substack{I \oplus I' \phantom{ \oplus Q (I \oplus I')} \\ \cap \phantom{ ' \oplus Q (I \oplus I')}  \\ I \oplus I' \oplus Q (I \oplus I')}} e^{(\Sfr{I}+\Sfr{I'})/\hbar} \otimes\id_{R'}$ by Proposition \ref{prop:Fubini}.
\end{proof}

\subsubsection{Fiber Integrals and Homological Perturbation Theory}\label{ssec:fibintHPL}

We finish this section by showing that the axiomatic definition of the perturbative integral can be easily connected to the  homological perturbation lemma (see e.g.\ \cite{crainic:hpl}). This argument first appeared in the Bc. thesis of O. Skácel \cite{OndraBC}; the construction of  perturbative BV integrals using homological perturbation lemma was anticipated in \cite[Remark~3]{cattaneo_mnev:chern-simons_invariants} and appeared explicitly in e.g.\  \cite{Albert, costello, gwilliam-thesis}; see also \cite[Sec.~5] {doubek_jurco_pulmann:quantum_L_infty_and_HPL} for a review.

\begin{lemma}[Uniqueness of the projection in a SDR] \label{lemma:uniquenessproj}Let $(i,p,k)$ be a SDR between $(V,Q_V)$ and $(W,Q_W)$. Then any chain map $p': V \to W$ satisfying $p'i=\id$ and $p'k=0$ is necessarily equal to $p$.
\end{lemma}
\begin{proof} We post-compose $ip = \id_V + Q_V k + k Q_V $ with $p'$ to get
\begin{equation}
     p'ip = p' + p'Q_V k + p'k Q_V
\end{equation}
By the assumptions on $p'$, the LHS equals $p$ while the last two terms on the RHS vanish.
\end{proof}
Recall from Proposition \ref{prop:SDRnondeg} that a non-degenerate reduction $L \colon V \dhxrightarrow{} R$ defines a symplectic SDR between $V$ and $R$. We extend it to a SDR $(I, P, K)$ between $\F{V}$ and $\F{R}$, cf. \cite[Sec.~3.3]{doubek_jurco_pulmann:quantum_L_infty_and_HPL}. Finally, we can see $\hbar\BV$ as a  perturbation of $\{\Sfree, -\}$, which allows us to use the homological perturbation lemma to perturb the other maps to get a new SDR. Namely for $P$, the perturbed projection is equal to
\begin{equation}\label{eq:Pprime}
    P' = P(1 + \hbar \BV K+ (\hbar \BV K)^2 + \dots ). 
\end{equation} 
\begin{prop}[{\cite[Sec.~3.2.3]{OndraBC}}]\label{prop:fiberBVisHPL}
    Let $L \colon V \dhxrightarrow{} R$ be a non-degenerate reduction. Then any (normalized) fiber integral along $L$ is necessarily equal to the map $P'$, obtained by deformation $\hbar \BV$ of the SDR induced by $L$ 
    \begin{equation}
        \frac{\int_L e^{\Sfr{I}/\hbar} f\rho}{\int_L e^{ \Sfr{I}/\hbar} \rho} = P'(f).
    \end{equation}
    Therefore, the BV integral from Definition \ref{def:BVintegralLagrangian} exists.
\end{prop}
\begin{proof}
    The first claim  follows directly by combining Proposition \ref{prop:fiber_integral} and  Lemma \ref{lemma:uniquenessproj}.

    For the existence statement, we want to prove that, for $\tilde{L}\subset V$ a nondegenerate Lagrangian, $P'$ satisfies the three conditions from Definition \ref{def:BVintegralLagrangian}. Namely, we need to check items \ref{i:dBVi:exact}, \ref{i:dBVi:vanish} and, instead of item \ref{i:halfdensnorm}, we need that $P'(1) = 1$. The first claim follows by the fact that $P'$ is a chain map between the perturbed differentials. The third claim follows easily from \eqref{eq:Pprime}, as $K(1) = 0$. Finally, the second claim, that $P'$ is zero on the vanishing ideal of $\tilde{L}$, is proven as follows: if we denote coordinates on $\tilde{L}$ by $\gamma$ and coordinates on $Q\tilde{L}$ by $\beta$, we have schematically $K \propto \beta \partial_\gamma$ and $\BV \propto \partial_\beta \partial_\gamma$. The composition $\BV K$ thus does not change the number of $\beta$'s in a monomial; and therefore $P'(\beta f) = P(1 + \hbar \BV K + \dots) (\beta f)=0$, as $P$ is zero on non-constant polynomials (see \cite[Sec.~4.1.1]{doubek_jurco_pulmann:quantum_L_infty_and_HPL} for a calculation of the tranferred differential and a more verbose version of this calculation).
\end{proof}

\section{Quantum \texorpdfstring{$(-1)$}{(-1)}-Symplectic Category}

 As explained by Ševera \cite{severa:qosc}, one should view Lagrangian submanifolds $\L \subset \M$ of an odd symplectic supermanifold $\M$ as distributional half-densities on $\M$. Indeed, the same way a half-density $\beta$ on $\M$ gives a functional\footnote{Provided the integral convergences, e.g.\ the body of $\M$ is compact.} on half-densities,
\begin{equation*}
    \alpha  \mapsto \intover{\M} \alpha \beta,
\end{equation*}
a Lagrangian $\L$ also gives a functional, a Dirac distribution supported on $\L$;
\begin{equation*}
    \alpha  \mapsto  \intover{\L} \alpha \equiv \intover{\M} \alpha \delta_\L .
\end{equation*}
This leads to a natural enlargement of (the odd version of) Weinstein's symplectic \say{category}: morphisms $\mathcal{M}_1 \to \mathcal{M}_2$ between $(-1)$-symplectic supermanifolds should be (distributional) half-densities on $\flip{\mathcal{M}_1}\times \mathcal{M}_2$ \cite[Def.~1]{severa:qosc}, with composition given by integration over the common factor.
\medskip

We now want to rigorously construct a linear version of such a category. That is, we would like the set of morphisms from $V$ to $W$ to contain both half-densities on $V\times W$ and Lagrangian subspaces of $\flip{V}\times W$. If we try to compose these two kinds of morphisms together, we get a diagram 
\begin{equation}
    \begin{tikzcd}
        * \arrow[rrr, "{ f\rho \in \D{V} }"] &&& V \arrow[rrr, "{L\subset_\text{Lagr.} \flip{V}\times W}"] &&& W.
        \end{tikzcd}
\end{equation}
It is natural to use the factorization of $L$ from Proposition \ref{prop:red_cored_factorization} to take the integral along the reduction
$V\dhxrightarrow{} R_V$ to get $\int_{\Ker{L}}f\rho \in \D{R_V}$. Moreover, the isomorphism $\phi\colon R_V \to R_W = \Im L / (\Im L)^\omega$ can be used to define
\begin{equation}  \phi_* \int_{\Ker{L} \smallsubset V} f \rho \in \D{(\Im L / (\Im L)^\omega)}. 
\end{equation}
This leads us to the following definition of a distributional half-density on $V$.
\begin{definition}\label{def:gen_Lagr}
    Let $( V, \omega)$ be a $(-1)$-shifted symplectic vector space. A \textbf{generalized Lagrangian} in $V$ is a triple $( C , f\rho , \Sfree )$ where
    \begin{itemize}
        \item $C \subseteq V$ is a coisotropic subspace,
        \item $f\rho \in \D{(C/C^\omega)} $ is a half-density on the coisotropic reduction,
        \item $\Sfree \in \Sym^2{((C/C^\omega)^*)}$  is a solution of the classical master equation on the coisotropic reduction, i.e.\ a differential on $C/C^\omega$ compatible with the symplectic form.
    \end{itemize}
    Given such generalized $( C , f \rho , \Sfree )$, we define a square-zero operator
    \[
        \hbar\BV ( C , f \rho , \Sfree ) \define (C, \hbar \BV(f) \rho + \{\Sfree, f\} \rho, \Sfree ). \qedhere
    \]
    \end{definition}
    
    \begin{remark}\label{rmk:distributional_halfdensity}
        Informally, such generalized Lagrangian should be seen as the ``distributional half-density'' \[e^{\Sfree/\hbar}f\rho \otimes \delta_{C^\omega}\] on $V$, using a (non-canonical) decomposition $V \cong (C/C^\omega)\oplus \sT{C^\omega}$. This also motivates the definition of the action of $\hbar \BV$ on generalized Lagrangians, see also \cite[Thm.~3]{severa:qosc}.

        Generalized Lagrangians can be understood as a \say{quantum} version of Lagrangian subspaces: Considering a half-density $e^{\Sfree/\hbar}$ and taking $\hbar\to 0$ limit, the path integral with weight $e^{\Sfree/\hbar}$ localizes to a Lagrangian subspace. For example, consider $\sT\mathbb R$ with even and odd coordinate denoted by $x$ and $\xi$. Then the distributional limit is 
        \vspace{-2mm}
        \[\lim_{\hbar=0} \hbar^{-1/2} e^{\frac {-1}{2} a x^2 /\hbar} \sqrt{dx d\xi} = \sqrt\frac{2\pi}{a} \delta_{x=0}\] (ignoring pairing with non-transversal $\delta_{x=0}$). See also the work of Albert Schwarz \cite[Sec. 7, Lemmata 8,8']{schwarz:semiclassical_in_BV}. 
    \end{remark}
We would like to define a category where morphisms $V \to W$ are generalized Lagrangians in $\flip{V}\times W$. To compose such morphisms, we need to investigate compositions of coisotropic relations in more detail.

\subsection{Coisotropic Relations}
A \textbf{coisotropic relation} from $V_1$ to $V_2$ is a coisotropic subspace $C \subseteq \flip{V_1} \times V_2$, see e.g.\ \cite{weinstein:coiso_and_poisson_groupoids, weinstein:coiso_relations}. Composition of coisotropic relations is defined by the usual composition of set-theoretic relations from equation \eqref{eq:comp}. Such composition is again coisotropic (see e.g.\ Remark \ref{rmk:XandL} below), so we have a category $\Coiso$ of coisotropic relations of $(-1)$-symplectic vector spaces. It will be useful now to denote the coisotropic reduction more concisely by
\begin{equation}
    \Red_C \define C/C^\omega .
\end{equation}
Let us now define a reduction $\Red_C \times \Red_{C'} \dhxrightarrow{} \Red_{C'\circ C}$ which will be used to define a composition of generalized Lagrangian relations, see also Remark \ref{rmk:XandL}.
\begin{lemma}\label{lemma:comp_R}
Let $C \subseteq \flip{V_1}\times V_2$ and $C' \subseteq \flip{V_2}\times V_3$ be two coisotropic relations. Then their $\Red$-\textbf{compositor}, defined as the graded linear relation 
\begin{equation}
    \begin{tikzcd}
\Red_C \times \Red_{C'} \arrow[r, "\Comp_{C,C'}"] & \Red_{C'\circ C},
\end{tikzcd}
\end{equation}
\begin{equation}\label{eq:compositor}
     \Comp_{C, C'}\define\{ ( [v_1, v_2], [v_2, v_3], [v_1, v_3] ) \in \flip{\Red_C \times \Red_{C'}}\times \Red_{C'\circ C} \mid (v_1, v_2)\in C, (v_2, v_3)\in C' \} ,
\end{equation}
is a reduction, i.e.\ a surjective Lagrangian relation.  Furthermore, if $C'' \in \flip{V_3}\times V_4$ is coisotropic, then the following diagram in $\LinOSC$ commutes.
\begin{equation}\label{eq:XXXX}\begin{tikzcd}[column sep = 3cm]
	{\Red_C\times \Red_{C'}\times \Red_{C''}} & {\Red_{C'\circ C} \times \Red_{C''}} \\
	{\Red_C\times \Red_{C''\circ C'}} & {\Red_{C''\circ C' \circ C}}
	\arrow["{\Comp_{C, C'}\times \diag{\Red_{C''}}}", from=1-1, to=1-2]
	\arrow["{\Comp_{C'\circ C, C''}}", from=1-2, to=2-2]
	\arrow["{\diag{\Red_C}\times \Comp_{C', C''}}"', from=1-1, to=2-1]
	\arrow["{\Comp_{C, C''\circ C'}}"', from=2-1, to=2-2]
\end{tikzcd}\end{equation}
\end{lemma}
\begin{proof}
    The relation $\Comp_{C, C'}$ is Lagrangian since it can be obtained by coisotropic reduction along
    \begin{equation} C \times C' \times (C'\circ C)\subset V_1\times \flip{V_2} \times V_2 \times \flip{V_3} \times \flip{V_1}\times V_3\end{equation}
    of the Lagrangian subspace
    \begin{equation} \diag{V_1\times V_2 \times V_3}, \end{equation}
   using Lemma \ref{lemma:reduced_lagrangian}.
    It is surjective since for any $(v_1, v_3) \in C'\circ C$, one can (by definition) find $v_2$ such that $(v_1, v_2)\in C$ and $(v_2, v_3)\in C'$.

    \smallskip
    Finally, both legs of the square \eqref{eq:XXXX} compose to relations $\Red_C\times \Red_{C'}\times \Red_{C''} \to \Red_{C''\circ C' \circ C}$ containing
    \begin{equation}\Comp_{C, C', C''}\define\{ ([v_1, v_2], [v_2, v_3], [v_3, v_4], [v_1, v_4])\mid  (v_1, v_2)\in C, (v_2, v_3)\in C', (v_3, v_4)\in C'' \}.  \end{equation}
    Since $\Comp_{C, C', C''}$ is Lagrangian (by a similar argument as above), for dimensional reasons (Lemma \ref{lemma:dimsum_symp}) the two legs of the square are necessarily equal to it.
\end{proof}
\begin{remark}\label{rmk:XandL}
The relation $\Comp_{C, C'}$ can be more abstractly constructed as the composition of the following Lagrangian relations
\[\begin{tikzcd}[column sep = 1.1in]
	{\flip{V_1} \times V_2 \times \flip{V_2}\times V_3} & {\flip{V_1}\times V_3} \\
	{\Red_C \times \Red_{C'}} & {\Red_{C'\circ C}}
	\arrow["{\mathrm{red}_{V_1 \times \diag{V_2} \times V_3}}", two heads, from=1-1, to=1-2]
	\arrow["{\mathrm{red}_{C\times C'}^T}", tail, from=2-1, to=1-1]
	\arrow["{\mathrm{red}_{C'\circ C}}", two heads, from=1-2, to=2-2]
	\arrow["{\Comp_{C, C'}}"', dashed, two heads, from=2-1, to=2-2]
	\arrow[dashed, from=2-1, to=1-2]
\end{tikzcd}\]
The diagonal arrow has $C'\circ C$ as its image, which proves that $C'\circ C$ is coisotropic. 
\medskip

Note also that composition of relations $V_1 \xrightarrow{L_1} V_2 \xrightarrow{L_2} V_3$ is given by reduction\footnote{Denoted $[L_1 \times L_2]_C$ in Lemma \ref{lemma:reduced_lagrangian}.} along the top line of the diagram, while ``composition'' of half-densities on $\Red_C$ and $\Red_{C'}$ is given by ``reduction'' (fiber integral) along the bottom line, see Definition \ref{def:qosc}.
\end{remark}

\begin{remark}\label{rmk:R_2functor}
The $\Red$-compositor $\Comp_{\bullet}$ defined by equation \eqref{eq:compositor} provides structure of a lax 2-functor on
\begin{equation}
    \Red_{\bullet} \colon \Coiso \to \B \! \OSC ,
\end{equation}
where the 2-category of coisotropic relations $\Coiso$ has only identity 2-cells, while $\B \! \mathcat{\LinOSC}$ is the one-object 2-category associated to the symmetric monoidal category $(\LinOSC, *, \times)$. This appears to be a part of a higher categorical structure involving coisotropic reductions and half-densities, which we will explore in future work.
\end{remark}

\subsection{Quantum Linear \texorpdfstring{$(-1)$}{(-1)}-Shifted Symplectic Category}\label{ssec:gen_Lagrangians}
We are now ready to define a category where morphisms are given by generalized Lagrangian relations. Since we can only compose morphisms if the appropriate perturbative integrals are well defined, we only get a \textbf{partial category}, where composition is not always defined.

\begin{definition}\label{def:qosc}
     The \textbf{quantum linear $(-1)$-symplectic category} $\QOSC$ is the partial category where:
    \begin{itemize}
        \item Objects are finite-dimensional $(-1)$-shifted symplectic vector spaces.
        \item Morphisms in $\QOSC \left( U , V \right)$ are generalized Lagrangians in $\flip{U} \times V$.
        \item The identity is given by the diagonal $ \left( \diag{V} \in \Coiso (V,V) , 1 , 0 \right)$.
    \end{itemize}
The composition of
\begin{equation}
    \begin{tikzcd}
        V_1 \arrow[rrr, "{( C ,  f\rho , \Sfree )}"] &&& V_2 \arrow[rrr, "{( C' , f'\rho ' , \Sfr{'} )}"] &&& V_3
        \end{tikzcd}
\end{equation}
is defined if $(\Sfree + \Sfr{'})$ is non-degenerate\footnote{In other words, if $\Comp_{C,C'}$ is a non-degenerate reduction from $\left( \Red_C \times \Red_{C'}, (\Sfree + \Sfr{'}) \right)$ to $\Red_{C'\circ C}$.} on $\Ker \Comp_{C,C'}$, and is given by
\begin{equation}\label{eq:qosc_comp}
    ( C' , f'\rho ' , \Sfr{'}) \after ( C ,  f\rho , \Sfree ) \define \big(  C' \after C,  \intover{\Comp_{C,C'}} e^{(\Sfree + \Sfr{'} )^{\ker \Comp_{C,C'}}/\hbar} f\rho \otimes f'\rho ' , (\Sfree+\Sfr{'})^{\Red_{C \after C'} } \big).
\end{equation}
Here, $\Comp_{C, C'}\colon \Red_C\times \Red_{C'} \dhxrightarrow{} \Red_{C'\circ C}$ is the $\Red$-compositor from Lemma \ref{lemma:comp_R}.
\end{definition}

\begin{prop}\label{prop:qosc_associative}
The composition of $\QOSC$ is unital and associative. Moreover, for two composable morphisms, we have
\begin{align*}
 & \hbar \BV [( C' , f'\rho ' , \Sfr{'}) \after ( C ,  f\rho , \Sfree )] \\ &= \hbar \BV( C' , f'\rho ' , \Sfr{'}) \after ( C ,  f\rho , \Sfree ) + (-1)^{\deg{f'}}( C' , f' \rho ' , \Sfr{'}) \after \hbar \BV ( C ,  f\rho , \Sfree ), \end{align*}
where the sum of two such generalized Lagrangians is defined by adding their half-density components.\footnote{It is possible to define such addition of generalized Lagrangians, if they have the same coisotrope and differential. This way, the category $\QOSC$ becomes enriched in the category of dg vector spaces, via the operator $\hbar \BV$.}
\end{prop}
\begin{proof}
The composition \eqref{eq:qosc_comp} is defined by transferring along the reduction $\Comp_{C, C'}$.  Composing $C'\circ C$ with $C'=\diag{V_2} $, we get $\Comp_{C, C'} = \diag{\Red_C} \colon \Red_{C} \to \Red_{C}$, which proves unitality. 

\smallskip
When considering general $(C''\circ C') \circ C$ and $C''\circ(C' \circ C)$, the resulting composite reductions are equal by diagram \eqref{eq:XXXX} of Lemma \ref{lemma:comp_R}. Thus, associativity for the composed differential follows from Appendix \ref{sec:transfer_of_Q}, while associativity for the composed half-density follows from Lemma \ref{lemma:fubinitransfer}.
\smallskip

Finally, the compatibility of the composition with the operator $\hbar\BV$ follows immediately from Item \ref{item:fiberintBVcomm} of Proposition \ref{prop:fiber_integral}. 
\end{proof}

\begin{remark}[Why is $\QOSC$ only a partial category?]
    In the symplectic category of smooth symplectic manifolds and smooth Lagrangian relations, the composition of two Lagrangians is defined only if the set-theoretic composition is smooth. The reason why $\QOSC$ is a partial category is \emph{different}: two morphisms are composable if the relevant BV integral converges. Since we define our integrals perturbatively, this is equivalent to invertibility of the quadratic part of the action; but one can imagine different contexts where the integrals are over e.g.\ compact manifolds, and always converge (for the price of introducing transversality considerations).
    \smallskip

    However, non-composability is a well-known feature of BV formalism, and is usually solved by considering closed integrands (morphisms) and deforming the Lagrangian. This suggests that our $\QOSC$ should be seen as a (partial) subcategory of  bigger dg category; the physical content of BV theories would be captured by the homology of this dg category. See also Section \ref{ssec:QLinfty}, where we see that postcomposing with non-degenerate Lagrangians does not change the homology of the differential given by $\Sfree$; we expect that the missing morphisms could do just that.
\end{remark}

\subsubsection{Examples}
\begin{example}
Each Lagrangian relation $L \colon V \to W$ gives a generalized Lagrangian $(L, 1, 0)$, where $1 \in \D{(L/L)} = \mathbb \Rhbar$ should be thought of as a scalar multiplying $\delta_L$ (see Remark \ref{rmk:distributional_halfdensity}).
This way, we get $\OSC$ as a wide subcategory of $\QOSC$, since one can easily verify that $( L' , 1 , 0) \after ( L , 1 , 0) = ( L' \after L , 1 , 0) $. 
\end{example}

\begin{example} A composition
\vspace{-1em}
    \begin{equation}
        \begin{tikzcd}
            * \arrow[rrr, "{( C , f\rho , \Sfree )}"] &&& V \arrow[rrr, "{( C' , f'\rho' , \Sfr{'} )}"] &&& *
            \end{tikzcd}
    \end{equation}
    gives a formal Laurent series as a result (see Remark \ref{rem:funczerodim}), i.e.\ defines a pairing of generalized Lagrangians in $V$. If we denote $\pi , \pi '$ the projections to the coisotropic reductions of $C$, $C'$, then  
     \begin{equation}
         \ker \Comp_{C,C'} = \Im \Comp_{C,C'}^T =  \pi \times \pi ' ( C \cap C' ) .
        \end{equation}
         For $(\Sfree + \Sfr{'} )$ non-degenerate on $\pi \times \pi ' ( C \cap C' )$, this formal Laurent series is computed as
        \begin{equation}
            \intover{\pi \times \pi ' ( C \cap C' ) \subset \Red_C \times \Red_{C'}} e^{(\Sfree + \Sfr{'} )/\hbar} f\rho \otimes f'\rho '  \ \in \Rhbar. 
        \end{equation}  
    The operator $\hbar \BV$ is self-adjoint with respect to this pairing, due to Proposition \ref{prop:qosc_associative}. 
    \medskip
    
    Let us also highlight the following special cases, which show how the category $\QOSC$ contains the standard BV integrals.   \begin{enumerate}
        \item The case when both generalized Lagrangians come from Lagrangian subspaces was considered in the previous examples. 
        \item If both generalized Lagrangians have support $C=V$, they are given by quadratic functions $\Sfree, \Sfree' \in \Sym^2{V^*}$ and half-densities $f\rho, f'\rho' \in \D{V}$. Their composition is given by the formal integral of the density $e^{\Sfree + \Sfree'}f f' \rho\rho'$ over $V$.
        \item If one of the generalized Lagrangians is of the form $(V, f\rho, \Sfree)$ and the other one is given by $(L, 1, 0)$ with $L\subset V$ Lagrangian, their pairing is given by the (formal) integral $\int_L e^{\Sfree}f\rho$.
    \end{enumerate}    
    \end{example}  
We can generalize the last item above to transfer half-densities along Lagrangian relations.
    
\begin{example}\label{ex:L_after_rho}
    Let us now return to the motivating example above Definition \ref{def:gen_Lagr}. That is, we want to compose
    \begin{equation}
        \begin{tikzcd}
            * \arrow[rrr, "{( U, f\rho, \Sfree )}"] &&& U \arrow[rrr, "{( L ,  1 , 0 )}"]  &&& V.
            \end{tikzcd}
    \end{equation}
     Let us consider the factorization cospan of $L$ (Definition \ref{def:factor_cospan}).
    \begin{equation}
       \begin{tikzcd}
U \arrow[rd, "L_U"', two heads] \arrow[rr, "L"] &   & V \arrow[ld, "L_V", two heads] \\
                                                & R &                               
\end{tikzcd}
        \end{equation}
        Clearly, $L \after U = \Im L \in \Coiso (*,V)$. From the definition, it is easy to see that
    \begin{equation}
     \Comp_{U,L } = \{ ( [0, u], [u, v], [0, v] ) \mid  (u, v)\in L \} \subset \flip{U \times *}\times R,
    \end{equation}
    so $\ker \Comp_{U,L} = \ker L$ and $\Im \Comp_{U,L} = R$. Therefore
    \begin{equation}
        ( L , 1 , 0) \after ( U , f\rho , \Sfree ) = \big( \Im L ,  \intBV{L_U}{\Ker L} f\rho  , \Sfr{R}  \big).
    \end{equation}
\end{example}

\subsection{Quantum \texorpdfstring{$\Linfty$}{L-infinity} Algebras} \label{ssec:QLinfty}

Recall the definition of a quantum $\Linfty$ algebra from \cite{zwiebach:closed_string}, we will use the form \cite[Def.~7]{doubek_jurco_pulmann:quantum_L_infty_and_HPL}.
\begin{definition}\label{def:qLinfty}
    A \textbf{quantum $L_{\infty}$ algebra} structure on a $(-1)$-shifted symplectic space $\left(V, \omega\right) $ is defined by a sequence of elements 
    \[ \left\{ S^g_n \in \Sym^n{V^*} \mid n \ge 1, g \ge 0, 2g + n \ge 2\right\} \]
    such that the formal sum
    $$S = \sum_{n, g} S_n^g \hbar^g \in \F{V}$$
    satisfies the \emph{quantum master equation}
    \[\hbar \BV e^{S/\hbar} = 0. \qedhere \]
\end{definition}
Let us denote $\Sfree \define S^0_2$ and $\Sint = S-\Sfree$, interpreted as the free and the interaction parts of $S$ respectively. As a consequence of the quantum master equation, $(V, \omega, \Sfree)$ is a dg $(-1)$-symplectic vector space, since $\{ \Sfree, \Sfree\} = 0$. With the decomposition $S = \Sfree + \Sint$, the quantum master equation can be equivalently written as
\begin{equation} \label{eq:decQME} (Q + \hbar \BV) e^{\Sint/\hbar}=0 \qquad \text{ or } \qquad 
        \frac12 \left\{ \Sint , \Sint \right\} + (Q + \hbar \BV) \Sint = 0,\end{equation}
where we denote $Q = \{\Sfree, -\}$.

\begin{prop}\label{prop:qLinfty_as_a_morphism}
Let $S \in \F{V}$ be a quantum $\Linfty$ algebra on a $(-1)$-symplectic vector space. Then for any linear half-density $\rho \in \HalfDens{V}$, the triple $(V, e^{\Sint/\hbar}\rho, \Sfree)$ defines a $\hbar \BV$-closed morphism
\[\begin{tikzcd}[column sep = 1.2in]
	{*} & V
	\arrow["{(V, e^{\Sint/\hbar}\rho, \Sfree)}", from=1-1, to=1-2]
\end{tikzcd}\]
in the category $\QOSC$.
\end{prop}
\begin{proof}
The fact that the morphism is $\hbar\BV$-closed follows from the first form of the ``decomposed'' quantum master equation in \eqref{eq:decQME}.
\end{proof}
Thus, the same way as we could understand Lagrangian subspaces of $V$ as generalized points $L: * \to V$ in $\LinOSC$, quantum $\Linfty$ algebras give additional  generalized (dg) points of $V$ in $\QOSC$. Finally, we can now interpret the construction of the effective action\footnote{See \cite[Sec.~5]{doubek_jurco_pulmann:quantum_L_infty_and_HPL} for a review of other constructions of effective actions in \cite{costello, mnev, ChuangLazarevMinModel, BarannikovSolving, braun_maunder:minimal_models}} \cite{doubek_jurco_pulmann:quantum_L_infty_and_HPL} as a composition in $\QOSC$.
\begin{prop}\label{prop:qLinfty_composed_with_L}
    Let $S$ be a quantum $\Linfty$ algebra on $V$ and let $V \dhxrightarrow{L} R$ be a non-degenerate reduction with respect to $\Sfree$. Then the composition (see Example \ref{ex:L_after_rho})
\begin{equation} \label{eq:composition}\begin{tikzcd}[column sep = 1.2in]
	{*} & V & R
	\arrow["{(V, e^{\Sint/\hbar}\rho, \Sfree)}", from=1-1, to=1-2]
	\arrow["{(L, 1, 0)}", from=1-2, to=1-3]
\end{tikzcd}\end{equation}  
is a $\hbar\BV$-closed generalized Lagrangian $* \to R$ of the form $(R, e^{W/\hbar}\rho_R, \Sfr{R})$ such that $\Sfr{R} + W$ defines a quantum $\Linfty$ algebra on $R$. 
\end{prop}
\begin{example}
    In particular, the decomposition $V \cong H \oplus \Im Q \oplus C$ in \cite[Lemma~4]{doubek_jurco_pulmann:quantum_L_infty_and_HPL} automatically gives a symplectic SDR (Proposition \ref{prop:SDRnondeg}) and thus such a choice induces a non-degenerate reduction $L_{H} = \diag{H}\times C  \colon V \dhxrightarrow{} H$. Using Proposition \ref{prop:fiberBVisHPL} we get that the perturbed map $P_1$ from \cite[Sec.~4.1.1]{doubek_jurco_pulmann:quantum_L_infty_and_HPL} is equal (up to normalization) to the post-composition by the Lagrangian relation $L_H$ in $\QOSC$.
\end{example}
\begin{proof}[Proof of Proposition \ref{prop:qLinfty_composed_with_L}] From Example \ref{ex:L_after_rho}, we know the composition is given by the perturbative fiber integral of $e^{\Sint/\hbar}\rho$ along $\ker L$. The compatibility of composition with $\hbar\BV$ from Proposition \ref{prop:qosc_associative} implies that the resulting half-density is again $\hbar\BV$-closed. The half-density component of the composite \eqref{eq:composition} comes with a $\Rhbar$ factor from Proposition \ref{prop:fiber_integral} which can be absorbed into the linear half-density $\rho_R$ and the rest can be written as $e^{W/\hbar}$ for $W/\hbar \in \F{R}$ of \emph{weight at least $1$}, since the fiber integral is weight-homogeneous.

\smallskip
Thus, it remains to show that $W$ has only non-negative powers of $\hbar$ to conclude it defines a quantum $\Linfty$ algebra structure on $R$. As a consequence of Wick's Lemma \ref{lemma:wick}, we can use a standard argument for Feynman graphs: the function $e^{W/\hbar}$ is given by a sum over all graphs, and its logarithm $W/\hbar$ is given by a sum over all connected graphs $\Gamma$, each weighted by $\hbar^{\textnormal{genus}(\Gamma)-1}$.
\end{proof}

\subsection{Relations of Quantum \texorpdfstring{$\Linfty$}{L-infinity} Algebras}\label{sec:relsqlinfty}
Finally, we can now use the category $\QOSC$ to discuss possible notions of morphisms between quantum $\Linfty$ algebras. Since we can encode a quantum $\Linfty$ algebra on $V$ into a morphism $* \to V$, a natural candidate for a morphism $(U, S^U) \to (V, S^V)$ is a commutative triangle of the form:
\begin{equation}\label{eq:triangleunderpoint}
\begin{tikzcd}
                           & * \arrow[ld, "{ ( U  ,  e^{\Sint^U /\hbar}\rho_U , \Sfr{U} )}"'] \arrow[rd, "{ ( V  , e^{\Sint^V/\hbar} \rho_V , \Sfr{V} )}"] &   \\
U \arrow[rr, "{(C,f\rho,Q)}"' ] &                                                                            & V
\end{tikzcd}
\end{equation}
If ${(C,f\rho,Q)} = (L, 1, 0)$ for a Lagrangian relation $L \colon U \to V$, then $L$ has to be surjective and we get that $S^V$ is the effective action computed by the fiber integral along $L$. We will now generalize this to a Lagrangian relation $L$ where possibly $\Im L \subsetneq V$.

\begin{definition}\label{def:rel_of_qLinfty}
    Let $S^U=\Sfr{U}+\Sint^U$ and $S^V=\Sfr{V}+\Sint^V$ be quantum $\Linfty$ algebras on $U$ and $V$. We say a Lagrangian relation $L \colon U \to V$ is a \textbf{relation of quantum} $\Linfty$ \textbf{algebras} and write 
    \[S^U \stackrel{L}{\sim} S^V\]
    if the morphisms in the following diagram in $\QOSC$ are composable and the square commutes
    \begin{equation}
\begin{tikzcd}
                                   & * \arrow[rd, "{ ( V , e^{\Sint^V/\hbar}\rho_V , \Sfr{V} )}"] \arrow[ld, "{( U ,  e^{\Sint^U/\hbar} \rho_U , \Sfr{U} )}"'] &                                    \\
U \arrow[rd, "{( L_U , 1 , 0 )}"'] &                                                                                       & V \arrow[ld, "{ ( L_V , 1 , 0 )}"] \\
                                   & R                                                                                   &                                   
\end{tikzcd}
\end{equation}
    for some choice of linear half-densities $\rho_U \in \HalfDens{U}$, $\rho_V \in \HalfDens{V}$. The Lagrangian relations $L_U$, $L_V$ are the factorization cospan of $L$ from Definition \ref{def:factor_cospan}, i.e.\ they are reductions such that $L= L_V^T \after L_U$.
    \end{definition}
    
    \noindent
       Unraveling the definition, a relation of quantum $\Linfty$ algebras satisfies the following:
      \begin{enumerate}
         \item The kernels $ \ker L = \ker L_U \subset U$ and $\ker L^T = \ker L_V \subset V$ are non-degenerate isotropes.
         \item The two differentials transferred along $L_U$ and $L_V$ to $R$ coincide.
         \item For some linear half-densities $\rho_U$ and $\rho_V$, 
         \[{\displaystyle\intBV{L_U}{\Ker L} e^{\Sint^U/\hbar}  \rho_U =
         \intBV{L_V}{\ker L^T} e^{\Sint^V/\hbar}  \rho_V}.\]
     \end{enumerate}

\begin{remark}\label{rmk:cospans_in_under_cat}
    A relation of quantum $\Linfty$ algebras can be described as a cospan in $* / \QOSC$.

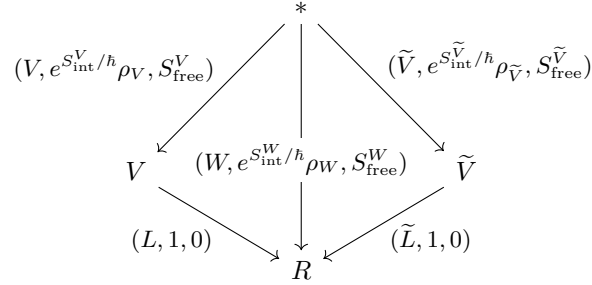
\[\begin{tikzcd}[ampersand replacement=\&,sep=2.25em]
	\&\& {*} \\
	\\
	V \&\&\&\& {\tilde{V}} \\
	\&\& R
	\arrow["{{ ( V , e^{\Sint^V/\hbar}\rho_V , \Sfr{V} )}}"',  from=1-3, to=3-1]
	\arrow["{{ ( \tilde{V} , e^{\Sint^{\tilde{V}}/\hbar}\rho_{\tilde{V}} , \Sfr{\tilde{V}} )}}", from=1-3, to=3-5]
	\arrow["{{ ( W , e^{\Sint^W/\hbar}\rho_W , \Sfr{W} )}}"{description, pos=0.6}, from=1-3, to=4-3]
	\arrow["{{( L , 1 , 0 )}}"', from=3-1, to=4-3]
	\arrow["{{(\smash{\tilde{L}}, 1 , 0 )}}", from=3-5, to=4-3]
\end{tikzcd}\]

\end{remark}

\subsubsection{Composing Relations of Quantum \texorpdfstring{$\Linfty$}{L-infinity} Algebras}

It is natural to ask whether relations of quantum $\Linfty$ algebras form a category; can they always be composed? We formulate a sufficient condition: they are composable when the \emph{underlying factorization cospans compose along pushouts} as in Corollary \ref{thm:cospans_comp}.

\begin{theorem}\label{thm:qLinfty_relation_composition}
            Let $S^U \stackrel{L_1}{\sim} S^V$ and $S^V \stackrel{L_2}{\sim} S^W$. If, moreover, $L_1$ and $L_2$ compose orthogonally, then
            \begin{equation}
                 S^U \stackrel{L_2 \after L_1}{\sim} S^W.
            \end{equation}
    \end{theorem}
    \begin{proof} We will prove that $L_2 \after L_1$ satisfies Definition \ref{def:rel_of_qLinfty}. Consider the diagram from Theorem \ref{thm:cospans_comp} given by factorization cospans of $L_1$, $L_2$ and $\tilde{L} \after L^T$. It commutes by the orthogonality assumption.
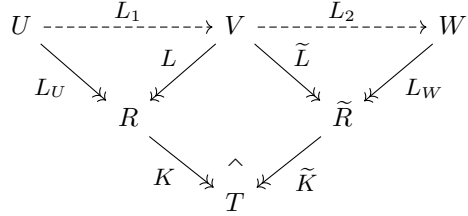
\begin{equation}\label{diag:comp_diagram}
\begin{tikzcd}
U \arrow[rd, "L_U"', two heads] \arrow[rr, "L_1", dashed]  &                               & V \arrow[rd, "\tilde{L}", two heads] \arrow[rr, "L_2", dashed] \arrow[ld, "L"', two heads] &                                                  & W \arrow[ld, "L_W", two heads] \\
                                                                                                      & R \arrow[rd, "K"', two heads] &                                                                                    & \widetilde{R} \arrow[ld, "\tilde{K}", two heads] &                                                                                            \\
                                                                                                      &                               & T \arrow[uu, phantom, "\lrcorner"{anchor=center, pos=0.125, rotate=135}]                                                                                 &                                                  &                                                                                           
\end{tikzcd}
\end{equation}
Note that the factorization cospan of $L_2 \after L_1$ is $K \after L_U$, $\tilde{K} \after L_W$.
    \begin{enumerate}
\item First, we prove non-degeneracy of $K \after L_U$ (the case of $\tilde{K} \after L_W$ is completely analogous). Denote $\ker L_U = I_U$, $\ker L = I$, $\ker \tilde{L} = \tilde{I}$, $\ker (K \after L_U) = J$. Let $u \in J \cap (Q_U J)^{\omega_U}$, we need to prove that then $u=0$. It is enough to prove that $u \in I_U$, since the case of $u \in I_U \cap (Q_UJ)^\omega$ is trivial.
The idea is to \say{transport} the property $u \in (Q_U J)^\omega$ to $V$ using $u \in J = \ker (L_2 \after L_1)$, where it becomes $v \in \smash{\tilde{I}} \cap (Q_V \smash{\tilde{I}})^\omega$ for some $v \in V$ such that $u \sim_{L_1} v$. Non-degeneracy of $\smash{\tilde{I}}$ implies that $v=0$, thus $u \sim_{L_1} 0$, i.e.\ $u \in I_U$. For details, see the proof of Lemma \ref{lemma:comp_nondeg} of Appendix \ref{ssec:comp_nondeg}.

\item Thanks to non-degeneracy condition verified above, the transfer of $Q_{U}$ and $Q_{W}$ to $T$ are well-defined. By the assumptions $\Sfr{U} \stackrel{\smash{L_1}}{\sim} \Sfr{V}$ and $\Sfr{V} \stackrel{\smash{L_2}}{\sim} \Sfr{W}$, the differentials can equivalently be transferred from $V$ along $K \after L$ and $\tilde{K} \after \tilde{L}$. But by the orthogonality condition and Corollary~\ref{thm:cospans_comp}, $\tilde{K} \after \tilde{L} = K \after L$ and the transferred differentials coincide.

\item By the above arguments, the perturbative BV integrals along $K \after L_U$ and $\tilde{K} \after L_W$ are well-defined. Using Lemma \ref{lemma:fubinitransfer}, $\smash{S^U \stackrel{L_1}{\sim} S^V}$, and $\smash{S^V \stackrel{L_2}{\sim} S^W}$, we can repeat the argument from the previous step and conclude the integrals also coincide.
     \qedhere
    \end{enumerate}   
    \end{proof}

\begin{remark}[On the orthogonality assumption.]
The assumption that $L_1$ and $L_2$ compose orthogonally is a convenient sufficient condition for composition of relations: The two effective actions on $T$ in \eqref{diag:comp_diagram} are given by integrating $e^{S^V/\hbar}$ along the left and right leg of the square (through $R$ and $\widetilde{R}$); and since the square commutes by the orthogonality assumption, the effective actions are equal for any $S^V$. A mild relaxation of the orthogonality condition is requiring the square in \eqref{diag:comp_diagram} to commute  up to homotopy; this would give homotopic actions on $T$ (i.e. equal in a homotopy category or equal up to a nonlinear change of coordinates \cite[Thm.~5]{doubek_jurco_pulmann:quantum_L_infty_and_HPL}).
\smallskip

Another interesting case would be an ``accidental'' equality of the effective actions on $T$, even when the two legs of the square (and thus the BV fiber integrals) are not equal. This would only work for some $S^V$, and could be physically more interesting; we do not know any examples of this kind.
\end{remark}

        \begin{remark}\label{rmk:constructiongspans}
        Instead of our definition of a relation of quantum $\Linfty$ algebras (see Remark \ref{rmk:cospans_in_under_cat}) we may consider a span in $*/\QOSC$ of the form
\[\begin{tikzcd}[row sep=2.2em]
	&& {*} \\
	\\
	&& W \\
	V &&&& {\tilde{V}}
	\arrow["{( L , 1 , 0 )}", near end, from=3-3, to=4-1]
	\arrow["{(\smash{\tilde{L}}, 1 , 0 )}"', near end, from=3-3, to=4-5]
	\arrow["{ ( V , e^{\Sint^V/\hbar}\rho_V , \Sfr{V} )}"', bend right, from=1-3, to=4-1]
	\arrow["{ ( W , e^{\Sint^W/\hbar}\rho_W , \Sfr{W} )}"{description}, pos=0.7, from=1-3, to=3-3]
	\arrow["{ ( \tilde{V} , e^{\Sint^{\tilde{V}}/\hbar}\rho_{\tilde{V}} , \Sfr{\tilde{V}} )}", bend left, from=1-3, to=4-5]
\end{tikzcd}\]      
 to be a morphism between $S^V = \Sfr{V}+\Sint^V$ and $S^{\tilde{V}}= \Sfr{\tilde{V}}+\Sint^{\tilde{V}}$. If we require the span of reductions $(L,\tilde{L})$ to be orthogonal, by Theorem \ref{thm:qLinfty_relation_composition} (taking $L_1 = L^T $, $L_2 = \tilde{L}$) we also have
        \[S \stackrel{\tilde{L} \after L^T}{\sim} \tilde{S}.\]
        So orthogonal spans in $*/\QOSC$ are special cases of relations of quantum $\Linfty$ algebras. The opposite problem—associating an orthogonal span of relations of quantum $\Linfty$ algebras to a relation (i.e.\ cospan) of quantum $\Linfty$ algebras—is much more difficult and it appears it poses the need for (formal) non-linear generalization of the linear quantum $(-1)$-symplectic category. We will explore this problem in future work. 
        \end{remark}

\appendix

\section{Appendix}

\subsection{Inductive Construction of a Complement to a Coisotrope}\label{ssec:inductive_construction_of_B}
Let us prove a lemma used to construct decompositions of $V$ in Section \ref{ssec:coisotropesdec}.
\begin{lemma}\label{lemma:inductive_construction_of_B}
    Let $C \subseteq V $ be a coisotropic subspace. Denote $I \define C^{\omega} \subseteq C$ its symplectic complement, which is isotropic. Then, for every $n \in \{ 0,1,\ldots \}$ there exists an isotropic complement 
    \begin{equation}B^{(n)} = B_n \oplus B_{1-n} \subset V_n \oplus V_{1-n}\end{equation}
    of $C_n \oplus C_{1-n}$ satisfying
    \begin{align*}
        \dim B_n & = \dim I_{1-n} , \\
        \dim B_{1-n} & = \dim I_{n}.
    \end{align*}
    In other words, the $\Z$-graded vector space\footnote{For infinite-dimensional $V/C$, this requires the axiom of countable choice.}
    $\smash{B \define \bigoplus\limits_{k \geq 0} B^{(n)} \subset V}$ is an isotropic complement of $C$
    satisfying
    \begin{equation}
        s^{-1}\dimsum{I}{s}  =  \dimsum{B}{s^{-1}} \quad \text{or equivalently} \quad       s^{-1}  \dimsum{B}{s}  = \dimsum{I}{s^{-1}}. 
    \end{equation}
\end{lemma}
\begin{proof}
       To simplify the notation, let us fix $n \geq 0$ and drop the superscript $(n)$ in $B^{(n)}$. We will work by induction on $\dim B_n + \dim B_{1-n}$. The induction hypothesis will be that there is a graded subspace $B \subseteq V$ such that 
        \begin{itemize}
            \item $B$ is isotropic,
            \item $B \cap C = 0$,
            \item $ \dim B_k = 0$ for all $ k \notin   \{ n , 1-n \}$,
            \item $ \dim B_k \leq \dim I_{1-k}$ for all $ k \in \{ n , 1-n \} $.
        \end{itemize} 
        The induction starts with $B = \{0\}$. In each step, if $ \dim B_k < \dim I_{1-k}$ for some $k \in \{ n , 1-n \}$, we will choose an element 
        $b\in \left( B^\omega \setminus \left( B \oplus C \right) \right)_k$ and change $B$ to $B' = \Span{B , b}$. This new $B'$ again satisfies the four properties above; once we reach $ \dim B_k = \dim I_{1-k}$ for all $ k \in \{ n , 1-n \} $, the induction stops. 

        To show that such $b$ exists, we will show that
        \begin{equation}\label{eq:dimineq}
            \dim \left( B^\omega \right)_k - \dim \left( B^\omega \cap \left( B \oplus C \right) \right)_k =  \dim I_{1-k} - \dim B_k,
        \end{equation} 
        and, provided $ \dim B_k < \dim I_{1-k}$, we can find a suitable $b$.

        To prove \eqref{eq:dimineq}, we will use the fact that $B^\omega \cap \left( B \oplus C \right) = \left( B^\omega \cap C \right) \oplus B$ and by Lemma \ref{lemma:ortho_comp_squared}, the symplectic complement is an involution, so $B^\omega \cap C = \left( B + I \right)^\omega =  \left( B \oplus I \right)^\omega $ and
        \begin{equation}
            \dimsum{B^\omega}{s} - \dimsum{ B^\omega \cap \left( B \oplus C \right) }{s} = \dimsum{B^\omega}{s} - \dimsum{\left( B \oplus I \right)^\omega}{s} - \dimsum{B}{s}.
        \end{equation}
        Now we use Lemma \ref{lemma:dimsum_symp}, which says that $\dimsum{W^\omega}{s} = s \dimsum{V}{s^{-1}} - s \dimsum{W}{s^{-1}}$. Four terms cancel out and we are left with
        \begin{equation}
            \dimsum{B^\omega}{s} - \dimsum{B^\omega \cap \left( B \oplus C \right) }{s} = s \dimsum{I}{s^{-1}} - \dimsum{B}{s}.
        \end{equation}
        The coefficient at $s^k$ of this equation is exactly \eqref{eq:dimineq}.
\end{proof}

\subsection{Transporting Differentials Along Reductions}
 \label{sec:transfer_of_Q}

\begin{prop}
	Let $C\subset V$ be a coisotropic subspace and let $Q$ be a differential on $V$ such that $L: V \to R$ is a non-degenerate reduction (Definition \ref{def:nondeg}). Define a relation $Q^R \colon R \to \shift{1}R$ by as the composition of the following relations
\begin{equation}
    \begin{tikzcd}
R \arrow[r, "L^T", tail] & V \arrow[r, "\graph{Q}"] & \shift{1}V \arrow[r, "L", two heads] & \shift{1}R.
\end{tikzcd}
\end{equation}
	Then $Q^R$ is the graph of a degree 1 differential on $R$ and agrees with the map
 \begin{equation}
     \begin{tikzcd}[column sep=0.44in]
R \arrow[r, "\mathrm{incl}_R", hook] &  V \arrow[r, "Q"] & V \arrow[r, "\mathrm{proj}_R"] & R
\end{tikzcd}
 \end{equation}
	given by the canonical decomposition \eqref{eq:canonicaldecomposition}.
\end{prop}
\begin{proof}
	The relation $Q^R$ consists of pairs $\{([c], [Qc]) \in \flip{R} \times \shift{1}R \mid c\in C \text{ such that } Qc \in C\}$. We will first show it is coinjective and cosurjective.

\begin{itemize}
    \item Elements of $\Ker (Q^R)^T$ are of the form $[Qi]$ for $i\in I = \ker L$ such that $Qi \in C = I^\omega$. This is equivalent to $\omega(Qi,-)_I = 0$, i.e.\ $\Sfr{I}(i, -) = 0$, which implies $i = 0$ and thus $\Ker (Q^R)^T = 0$.    
    \item To show $\Im (Q^R)^T = R$, we want to show that each $r\in R$ has a representative $c\in C$ such $Qc \in C$. Let us choose any representative $c_0\in C$ of $r$. We are looking for $i_0 \in I$ such that $Q(c_0 + i_0)\in C = I^\omega$, i.e.\ for $i_0$ solving
	\begin{equation} \omega(-Qc_0, -)_I = \omega(Q i_0, -)_I = \Sfr{I}(i_0, -) \end{equation}
	which is possible as $\Sfr{I}$ is non-degenerate.
\end{itemize}

	To check that the induced map squares to $0$, the composition of relations $Q^R\circ Q^R$ is given by
	\begin{equation} [c] \sim c \sim Qc \sim [Qc] \sim Qc + i \sim Q^2 c + Qi \sim [Qi] \end{equation}
	but this is independent of the choice of $i$ (such that $Qi\in C$), i.e.\ we can take $i=0$.
 
 \smallskip	
    Finally, using the decomposition $V = I \oplus R \oplus QI$, we get that the relation $Q^R$ contains a pair $(r, Qr)$ coming from $r \sim (0, r, 0)\sim (0, Qr, 0) \sim Qr$.
\end{proof}

\subsection{Composition and Non-Degeneracy}
 \label{ssec:comp_nondeg}
Consider the diagram from Theorem \ref{thm:cospans_comp} given by factorization cospans of Lagrangian relations $L_1$, $L_2$ and $\tilde{L} \after L^T$. It commutes if and only if we assume $L_1$ and $L_2$ compose orthogonally, i.e.\ $\ker L \perp \ker \tilde{L}$.
\begin{equation}\label{diag:comp_diagramAppendix}
\begin{tikzcd}
U \arrow[rd, "L_U"', two heads] \arrow[rr, "L_1", dashed]  &                               & V \arrow[rd, "\tilde{L}", two heads] \arrow[rr, "L_2", dashed] \arrow[ld, "L"', two heads] &                                                  & W \arrow[ld, "L_W", two heads] \\
                                                                                                      & R \arrow[rd, "K"', two heads] &                                                                                    & \widetilde{R} \arrow[ld, "\tilde{K}", two heads] &                                                                                            \\
                                                                                                      &                               & T \arrow[uu, phantom, "\lrcorner"{anchor=center, pos=0.125, rotate=135}]                                                                                 &                                                  &                                                                                           
\end{tikzcd}
\end{equation}

\begin{lemma}\label{lemma:comp_nondeg}
    Let $L_1 \colon U \to V, L_2 \colon V \to W$ be Lagrangian relations between $(-1)$-shifted dg symplectic vector spaces $U,V,W$ (Definition \ref{def:dg_symp}) such that
    \begin{enumerate}
        \item the kernels  $I_U =\ker L_1, I= \ker L_1^T, \tilde{I} = \ker L_2$ are non-degenerate isotropes (Definition \ref{def:nondeg}),
        \item $\Sfr{U} \stackrel{L_1}{\sim} \Sfr{V}$ and $\Sfr{V}\stackrel{L_2}{\sim} \Sfr{W}$ (Definition \ref{def:rel_of_qLinfty}),
        \item $L_1$ and $L_2$ compose orthogonally (Definition \ref{def:ortho}).
    \end{enumerate}
    Then also the kernel of $L_2 \after L_1$ is a non-degenerate isotrope $J = \ker (L_2 \after L_1)$.
\end{lemma}

\begin{proof}
Let $u \in J \cap (Q_U J)^{\omega_U}$, we will prove that then $u=0$. It is enough to prove that $u \in I_U$, since by non-degeneracy of $I_U$ and $I_U \subseteq J $,
\begin{equation*}
    I_U \cap (Q_U J)^{\omega_U} \subseteq I_U \cap (Q_U I_U)^{\omega_U} = \{ 0 \}.
\end{equation*}
\textbf{Proof that $u \in I_U$.} From $I_U \subseteq J \subseteq J^{\omega_U} \subseteq I_U^{\omega_U} = I_U \oplus R_\can^U$, we have the decomposition $u = u_0 + u_R$ for $u_0 \in I_U$, $u_R \in R_\can^U \cap J$. By $u \in (Q_U J)^{\omega_U}$, for all $u' \in J$,
\begin{equation}
    \omega_U (u,Q_U( u'))=0.
\end{equation}
Decomposing $u' = u'_0 + u'_R \in I_U \oplus (R_\can^U \cap J)$, we have
\begin{equation*}
     \omega_U (u,Q_U (u'_0)) + \omega_U (u_0,Q_U (u'_R)) + 
     \omega_U (u_R,Q_U (u'_R)) = 0,
\end{equation*}
where the first term vanishes since $u \in  (Q_U I_U)^{\omega_U}$ and the second vanishes since $Q_U (R_\can^U) \subseteq R_\can^U$ and $R_\can^U \perp I_U$. We are thus left with the last term and using the fact that the projection $\pi_U : U \to R$ restricts to a symplectic isomorphism $\pi_U : R_\can^U \cong R$,
\begin{equation}
    \omega_R ( \pi_U (u_R),\pi_U \after Q_U (u'_R)) = 0.
\end{equation}

We will now prove that $(u_R,0) \in L_1$, i.e.\ $u_R \in \ker L_1 = I_U$, which implies that also $u \in I_U$. With this, the proof of non-degeneracy will be complete.

\smallskip
\noindent
\textbf{Proof that $u_R \in I_U$.}
By $u_R, u'_R \in J$, there exist $v, v' \in V$ such that $(u_R,v), (u'_R,v') \in L_1$ and $v, v' \in \ker L_2 = \tilde{I}$.  Using this, $\pi_U \after Q_U = Q_R^U \after \pi_U$, and $\Sfr{U} \stackrel{\smash{L_1}}{\sim} \Sfr{V}$, we obtain
\begin{equation}
    \omega_R ( \pi_U (u_R), Q_R^U \after \pi_U (u'_R)) = \omega_R ( \pi (v), Q_R^V \after \pi (v')) = 0.
\end{equation}
Since $v'$ is in the image of an arbitrary $u'_R \in J \cap R_\can^U$, the last equality holds for any 
\begin{equation}
    v' \in L_1 (J \cap R_\can^U) = L_1 (J) = \Im L_1 \cap \ker L_2 = I^{\omega_V} \cap \tilde{I}.
\end{equation}
By orthogonality, $I^{\omega_V} \cap \tilde{I} = \tilde{I}$, so we have
\begin{equation}\label{eq:nondeg_on_R}
    \omega_R ( \pi (v), Q_R^V \after \pi (v')) = 0 \quad \textnormal{ for any }v' \in \tilde{I}.
\end{equation}
Without loss of generality, we consider a representant $v \in \tilde{I} \cap R_\can^I$, where $ R_\can^I = (I \oplus Q_V I )^{\omega_V}$. Then we have, for any $v'=v'_0 + v'_R \in \tilde{I}$ with $v'_0 \in I$, $v'_R \in R_\can^I$,
\begin{equation}
\omega_V ( v , Q_V (v')) = \omega_V ( v , Q_V (v'_0))
+ \omega_V ( v , Q_V (v'_R)) = 0.
\end{equation}
The first term vanishes as $R_\can^I \perp Q_V I$, the second one contains only entries in $R_\can^I$, so it can be rewritten as $\omega_R ( \pi (v), \pi \after Q_V (v'_R)) = \omega_R ( \pi (v),  Q_R^V \after \pi (v'_R)) = \omega_R ( \pi (v),  Q_R^V \after \pi (v'))$ and it vanishes by equation \eqref{eq:nondeg_on_R}. We have proven that $v \in \tilde{I} \cap (Q_V \tilde{I})^{\omega_V}$ and by the non-degeneracy of $\tilde{I}$, we obtain $v=0$ and thus  $u_R \in I_U$.

\end{proof}

\subsection*{Data Availability and Conflicts of Interest Statement}

Data availability is not applicable to this article as no new data were created or analysed in this study. All authors declare that they have no conflicts of interest.

\begingroup
\let\clearpage\relax
\printbibliography[
heading=bibintoc,
title={Bibliography}
]
\endgroup

\end{document}